\documentclass[11pt]{article}
\usepackage{epsfig}
\usepackage{amssymb,amsmath, amsthm,url,graphicx,color}
\usepackage{multirow}
\usepackage{booktabs}

\usepackage{natbib}
\newtheorem{theorem}{Theorem}
\newtheorem{lemma}{Lemma}

\newtheorem{corollary}{Corollary}[section]

\numberwithin{equation}{section} \numberwithin{theorem}{section}
\numberwithin{lemma}{section} \numberwithin{remark}{section}
\numberwithin{example}{section}
\numberwithin{table}{section}
\numberwithin{figure}{section}

\newcommand{\bY}{\boldsymbol{Y}}
\newcommand{\bX}{\boldsymbol{X}}
\newcommand{\bR}{\boldsymbol{R}}
\newcommand{\bZ}{\boldsymbol{Z}}
\newcommand{\bE}{\boldsymbol{E}}
\newcommand{\be}{\boldsymbol{e}}
\newcommand{\bU}{\boldsymbol{U}}

\newcommand{\bbeta}{\boldsymbol{\beta}}

\newcommand{\bdelta}{\boldsymbol{\delta}}
\newcommand{\bDelta}{\boldsymbol{\Delta}}
\newcommand{\bu}{\boldsymbol{u}}
\newcommand{\hv}{\hat{v}}

\newcommand{\cov}{\mathrm{Cov}}

\newcommand{\tb}{\textcolor{blue}}

\usepackage{float}
\floatstyle{ruled}
\newfloat{algo}{thp}{lop}
\floatname{algo}{Algorithm}

\newcommand{\convP}{\stackrel{P}{\rightarrow}}

\allowdisplaybreaks

\begin{document}

\title{
On the prediction of stationary functional time series
}

\author{
Alexander Aue\footnote{Corresponding author}
\and Diogo Dubart Norinho
\and Siegfried H\"ormann
}

\date{}
\maketitle

\begin{abstract}
\setlength{\baselineskip}{2em}
This paper addresses the prediction of stationary functional time series. Existing contributions to this problem have largely focused on the special case of first-order functional autoregressive processes because of their technical tractability and the current lack of advanced functional time series methodology. It is shown here how standard multivariate prediction techniques can be utilized in this context. The connection between functional and multivariate predictions is made precise for the important case of vector and functional autoregressions. The proposed method is easy to implement, making use of existing statistical software packages, and may therefore be attractive to a broader, possibly non-academic, audience. Its practical applicability is enhanced through the introduction of a novel functional final prediction error model selection criterion that allows for an automatic determination of the lag structure and the dimensionality of the  model. The usefulness of the proposed methodology is demonstrated in a simulation study and an application to environmental data, namely the prediction of daily pollution curves describing the concentration of particulate matter in ambient air. It is found that the proposed prediction method often significantly outperforms existing methods.\medskip 
\end{abstract}

\setlength{\baselineskip}{2em}

\section{Introduction}
\label{s:1}

Functional data are often collected in sequential form. The common situation is a continuous-time record that can be separated into natural consecutive time intervals, such as days, for which a reasonably similar behavior is expected. Typical examples include the daily price and return curves of financial transactions data and the daily patterns of geophysical, meteorological and environmental data.  The resulting functions may be described by a time series $(Y_k\colon k\in\mathbb{Z})$, each term in the sequence being a (random) function $Y_k(t)$ defined for $t$ taking values in some interval $[a,b]$. Here, $\mathbb{Z}$ denotes the set of integers. The object $(Y_k\colon k\in\mathbb{Z})$ will be referred to as a functional time series (see \cite{hk12} for a recent survey on time series aspects, and \cite{fv10} and \cite{rs05} for general introductions to functional data analysis). Interest for this paper is in the functional modeling of concentration of particulate matter with an aerodynamic diameter of less than $10\mu m$ in ambient air, measured half-hourly in Graz, Austria. It is widely accepted that exposure to high concentrations can cause respiratory and related health problems. Local policy makers therefore monitor these pollutants closely. The prediction of concentration levels is then a particularly important tool for judging whether measures, such as partial traffic regulation, have to be implemented in order to meet standards set by the European Union.

Providing reliable predictions for future realizations is in fact one of the most important goals of any time series analysis. In the univariate and multivariate framework, this is often achieved by setting up general prediction equations that can be solved recursively by methods such as the Durbin-Levinson and innovations algorithms (see, for example, \cite{bd91,ss11}). Prediction equations may be derived explicitly also for general stationary functional time series (see Section 1.6 of the monograph \cite{b00}) but they seem difficult to solve and implement. As a consequence, much of the research in the area has focused on the first-order functional autoregressive model, shortly FAR(1). \cite{b00} has derived one-step ahead predictors that are based on a functional form of the Yule-Walker equations. \cite{bcs00} have proposed nonparametric kernel predictors and illustrated their methodology by forecasting climatological cycles caused by the El Ni\~{n}o phenomenon. While this paper, and also \cite{bc96}, have adapted classical spline smoothing techniques, \cite{as03}, see also \cite{aps06,aps09}, have studied FAR(1) curve prediction based on linear wavelet methods. \cite{ko08} have introduced the predictive factor method, which seeks to replace functional principal components with directions most relevant for predictions. \cite{dkz12} have evaluated several competing prediction models in a comparative simulation study, finding Bosq's (2000) method to have the best overall performance. Other contributions to the area are \cite{acv10}, and \cite{av08}.

In spite of its statistical relevance and its mathematical appeal, functional time series modeling
has still some unpleasant limitations for the practitioner. First, to date there are not many ``ready to use'' statistical software packages that can be utilized directly for estimation and prediction purposes. The only available packages that the authors are aware of are the {\tt far} package of \cite{dg10} and the {\tt ftsa} package of \cite{hs:12}, both implemented for the statistical software {\tt R}. The lack of tailor-made procedures often requires manual implementation. This may be challenging and therefore restrict use of the methodology to an academic audience.  Second, the methodology developed for  the FAR(1) case is difficult to generalize. If an FAR(1) approach is infeasible, one can use the multiple testing procedure of \cite{kr13} to determine an appropriate order $p$ for a more general FAR($p$) process. In addition, exogenous predictors can be incorporated using the work of \cite{dg02}. These authors include exogenous covariates of FAR(1) type into a first-order autoregressive framework for functional ozone predictions. For more general cases functional theory and estimation have not yet been developed.

The goal of this paper is then to fill in this gap by promoting a simple alternative prediction algorithm which consists of three basic steps, all of which are easy to implement by means of existing software. First, use functional principal components analysis, FPCA, to transform the functional time series observations $Y_1,\ldots,Y_n$ into a vector time series of FPCA scores $\bY_1,\ldots,\bY_n$ of dimension $d$, where $d$ is small compared to $n$. Second, fit a vector time series to the FPCA scores and obtain the predictor  $\hat\bY_{n+1}$ for $\bY_{n+1}$. Third, utilize the Karhunen-Lo\`eve expansion to re-transform $\hat\bY_{n+1}$ into a curve predictor $\hat Y_{n+1}$. The first and the third step are simple and can be performed, for example, with the {\tt fda} package in {\tt R}. The second step may be tackled with standard multivariate time series methodology. Details are developed in Section~\ref{s:pr}. While the proposed approach is conceptually quite easy, several non-trivial questions need to be raised:
\begin{enumerate}\itemsep-.5ex
\item How does the resulting method differ from existing ones?
\item Why is this method justified from a theoretical standpoint?
\item In order to minimize the prediction error, how can the number of principal components in the dimension reduction for Step~1 be determined and how should model selection be performed in Step~2? Preferably, both choices should be made simultaneously.
\end{enumerate}
These issues will be addressed in Section~\ref{s:ffar}. In particular, a comparison to Bosq's~(2000) classical benchmark FAR($p$) prediction is made. A theoretical bound for the prediction error of the proposed methodology is established, which will imply asymptotic consistency. In Section~\ref{ss:fpe} a novel functional final prediction error criterion is developed that jointly selects the order $p$ and the dimensionality $d$ of the FPC score vectors, thereby allowing for an automatic prediction process. 

Functional principal components have been employed in other approaches to functional prediction, for example in Bosq's~(2000) FAR(1) prediction method and in \cite{aov:1999}. Roughly speaking, these and many other existing approaches have in common that $Y_k$ is regressed onto the lagged observation $Y_{k-1}$ by minimizing $E[\int [Y_k(t)-\Psi(Y_{k-1})(t)]^2dt]$ with respect to a linear operator $\Psi$.  The solution of this problem involves an infinite series representation of $\Psi$ along FPCs. (More details will be given in Section~\ref{ss:ffar}.) In contrast, the proposed approach first uses dimension reduction via FPCA and then fits a model to the reduced data. No a priori knowledge of the functional model is needed and instead of a single estimator, a variety of existing tools for vector processes can be entertained. Further lags or exogenous covariates are also easily included into the prediction algorithm (see Section~\ref{s:covariates}). 

\cite{hu:07} and \cite{hs:09} have suggested a curve prediction approach based on modeling FPC scores by scalar time series. They argue that scores are uncorrelated and that hence individual time series can be fit. Depending on the structure of the data, this can be quick and efficient in some cases but less accurate in other cases. The fact that FPC score vectors have no instantaneous correlation, does not imply that autocovariances at lags greater than zero remain diagonal. Hence univariate modeling may invoke a loss of valuable information hidden in the dependence of the data. This will be demonstrated in Section~\ref{s:sim} as part of a simulation study. This issue can be avoided if one makes use of so-called {\em dynamic functional principal components} recently introduced in \cite{hkh13} and \cite{pt13}. These authors propose a methodology which produces score vectors with diagonal autocovariances via time invariant functional linear filters. Since the involved filters are two-sided (they require past and future observations), it is not clear how this methodology could be used for prediction.

It should be noted that, in this article, the data $Y_k$ are assumed to be given in functional form, since the focus is on working out functional prediction methodology without getting into aspects of data preprocessing, which appears to be rather specific to the particular data at hand and therefore not conducive to a unified treatment. In practice, however, one observes only vectors $Y_k(t_1),\ldots,Y_k(t_L)$, with spacings, $t_\ell-t_{\ell-1}$, and number of intraday sampling points, $L$, potentially varying from day to day. The problem of transforming the vector observations into (smooth) functions has been treated in many articles and will not be detailed here. As an excellent starting point for reading in this direction the reader is referred to Chapters~3--7 of \cite{rs05}. It is expected that the comparative results established in this paper as part of simulations and the application will hold also if the functions are not sampled equidistantly, with the rate of improvement of the proposed method over its competitors being of similar magnitude.

The remainder of the paper contains some possible extensions of the new prediction methodology in Section \ref{s:ex}, a supporting simulation study in Section \ref{s:sim} and an application to the prediction of intraday patterns of particulate matter concentrations in Section~\ref{s:appl}. Section \ref{s:con} concludes and technical proofs are given in Appendix \ref{s:comp}.

\section{Methodology}
\label{s:pr}

In what follows, let $(Y_k\colon k\in\mathbb{Z})$ be an arbitrary stationary functional time series. It is assumed that the observations $Y_k$ are elements of the Hilbert space $H=L^2([0,1])$ equipped with the inner product $\langle x,y\rangle=\int_0^1x(t)y(t)dt$.  Each $Y_k$ is therefore a square integrable function satisfying $\|Y_k\|^2=\int_0^1Y_k^2(t)dt<\infty$. All random functions  are
defined on some common probability space $(\Omega,\mathcal{A},P)$. The notation $Y\in L_H^p=L_H^p(\Omega,\mathcal{A},P)$ is used to indicate that, for some $p>0$, $E[\|Y\|^p]<\infty$. Any $Y\in L_H^1$ possesses then a mean curve $\mu=(E[Y(t)]\colon t\in[0,1])$, and any $Y\in L_H^2$ a covariance operator $C$, defined by $C(x)=E[\langle Y-\mu,x\rangle(Y-\mu)]$. The operator $C$ is a kernel operator given by
\[
C(x)(t)=\int_0^1c(t,s)x(s)ds,\qquad c(t,s)=\cov(Y(t),Y(s)).
\]
As in the multivariate case, $C$ admits the spectral decomposition
\[
C(x)=\sum_{\ell=1}^\infty\lambda_\ell\langle v_\ell,x\rangle v_\ell,
\]
where $(\lambda_\ell\colon\ell\in\mathbb{N})$ are the eigenvalues (in strictly descending order) and $(v_\ell\colon\ell\in\mathbb{N})$ the corresponding normalized eigenfunctions, so that $C(v_\ell)=\lambda_\ell v_\ell$ and $\|v_\ell\|=1$. Here, $\mathbb{N}$ is the set of positive integers. The $(v_\ell\colon\ell\in\mathbb{N})$ form an orthonormal basis of $L^2([0,1])$. Hence $Y$ allows for the Karhunen-Lo\'eve representation $Y=\sum_{\ell=1}^\infty\langle Y,v_\ell\rangle v_\ell$. The coefficients $\langle Y,v_\ell\rangle$ in this expansion are called the FPC scores of $Y$.

Suppose now that we have observed $Y_1,\ldots,Y_n$. In practice $\mu$ as well as $C$ and its spectral decomposition will be unknown and need to be estimated from the sample. We estimate $\mu$ by
\[
\hat\mu_n(t)=\frac{1}{n}\sum_{k=1}^nY_k(t),\qquad t\in[0,1],
\]
and the covariance operator by
\[
\hat{C}_n(x)=\frac{1}{n}\sum_{k=1}^n\langle Y_k-\hat{\mu}_n,x\rangle (Y_k-\hat{\mu}_n).
\]
Under rather general weak dependence assumptions these estimators are $\sqrt{n}$-consistent. One may, for example, use the concept of $L^p$-$m$-approximability introduced in \cite{hk11} to prove that
$E[\|\hat\mu_n-\mu\|^2]=O(n^{-1})$ and  $E[\|\hat{C}_n-C\|^2_\mathcal{L}]=O(n^{-1})$, where the operator norm $\|\cdot\|_\mathcal{L}$ is, for any operator $A$, defined by
\[
\|A\|_\mathcal{L}=\sup_{\|x\|\leq 1}\|A(x)\|.
\]
It is shown in Lemma~\ref{le:A1} of the Appendix that the general results (see Theorems~5 and 6 of \cite{hk12}) apply to the functional autoregressive processes studied in this paper. From $\hat{C}_n$, estimated eigenvalues $\hat\lambda_{1,n},\ldots,\hat\lambda_{d,n}$ and estimated eigenfunctions $\hat{v}_{1,n},\ldots,\hat{v}_{d,n}$ can be computed for an arbitrary fixed, but typically small, $d<n$. These estimators inherit $\sqrt{n}$-consistency from $\hat{C}_n$. See Theorem~7 in \cite{hk12}. For notational convenience, $\hat\lambda_{\ell}$ and $\hat{v}_{\ell}$ will be used in place of $\hat\lambda_{\ell,n}$ and $\hat{v}_{\ell,n}$.

Functional linear prediction equations for general stationary processes have been derived in Section 1.6 of the monograph \cite{b00}. They appear to be impractical for actual data analysis as there do not seem to be either articles discussing applications to real life examples or contributions concerned with further foundational elaboration. As pointed out in the introduction, the notable exception is the FAR(1) process defined by the stochastic recursion
\begin{equation}\label{eq:FAR_mu}
Y_k-\mu=\Psi(Y_{k-1}-\mu)+\varepsilon_k, \qquad k\in\mathbb{Z},
\end{equation}
where $(\varepsilon_k\colon k\in\mathbb{Z})$ are centered, independent and identically distributed innovations in $L_H^2$ and $\Psi:H\to H$ a bounded linear operator satisfying $\|\Psi^{k_0}\|_\mathcal{L}<1$ for some $k_0\geq 1$. The latter condition ensures that the recurrence equations \eqref{eq:FAR_mu} have a strictly stationary and causal solution in $L_H^2$. \cite{b00} has in the FAR(1) case used the prediction equations to devise what is now often referred to as the common predictor. This one-step ahead prediction is based on an estimator $\tilde\Psi_n$ of $\Psi$ and then given by $\tilde Y_{n+1}=\tilde\Psi_nY_{n}$. Details of this method are given in Section \ref{s:ffar}, where it will be used as a benchmark to compare with the novel methodology to be introduced in the following. The new prediction technique avoids estimating operators directly and instead utilizes existing multivariate prediction methods.

The proposed prediction algorithm proceeds in three steps. First, select $d$, the number of principal components to be included in the analysis, for example by ensuring that a certain fraction of the data variation is explained. With the sample eigenfunctions, empirical FPC scores $y_{k,\ell}^e=\langle Y_k,\hat v_\ell\rangle$ can now be computed for each combination of observations $Y_k$, $k=1,\ldots,n$, and sample eigenfunctions $\hat v_\ell$, $\ell=1,\ldots,d$. The superscript $e$ emphasizes that empirical versions are considered. Create from the FPC scores the vectors
\[
\bY_k^e=(y_{k,1}^e,\ldots,y_{k,d}^e)^\prime,
\]
where ${}^\prime$ signifies transposition. By nature of FPCA, the vector $\bY_k^e$ contains most of the information on the curve $Y_k$. In the second step, fix the prediction lag $h$. Then, use multivariate prediction techniques to produce the $h$-step ahead prediction
\[
\hat\bY_{n+h}^e=(\hat y_{n+h,1}^e,\ldots,\hat y_{n+h,d}^e)^\prime
\]
given the vectors $\bY_1^e,\ldots,\bY_n^e$. Standard methods such as the Durbin-Levinson and innovations algorithm can be readily applied, but other options such as exponential smoothing and nonparametric prediction algorithms are available as well. In the third and last step, the multivariate predictions are re-transformed to functional objects. This conversion is achieved by defining the truncated Karhunen-Lo\'eve representation
\begin{equation}\label{fc_kl}
\hat{Y}_{n+h}=\hat{y}_{n+h,1}^e\,\hat{v}_1+\cdots+\hat{y}_{n+h,d}^e\hat{v}_d
\end{equation}
based on the predicted FPC scores $\hat y_{k,\ell}^e$ and the estimated eigenfunctions $\hat v_\ell$. The resulting $\hat Y_{n+h}$ is then used as the $h$-step ahead functional prediction of $Y_{n+h}$. The three prediction steps are summarized in Algorithm~\ref{alg:1}.
\begin{algo}
\vspace{.6cm}
\caption{Functional Prediction}
\label{alg:1}
\begin{enumerate}
\item Fix $d$. For $k=1,\ldots,n$, use the data $Y_1,\ldots,Y_n$ to compute the vectors
\[
\bY_k^e=(y_{k,1}^e,\ldots,y_{k,d}^e)^\prime,
\]
containing the first $d$ empirical FPC scores $y_{k,\ell}^e=\langle Y_k,\hat{v}_\ell\rangle$.
\item Fix $h$. Use $\bY_1^e,\ldots,\bY_n^e$ to determine the $h$-step ahead prediction
\[
\hat\bY_{n+h}^e=(\hat{y}^e_{n+h,1},\ldots,\hat{y}^e_{n+h,d})^\prime
\]
for $\bY_{n+h}^e$ with an appropriate multivariate algorithm.
\item Use the functional object
\[
\hat{Y}_{n+h}=\hat{y}_{n+h,1}^e\,\hat{v}_1+\ldots+\hat{y}_{n+h,d}^e\hat{v}_d
\]
as $h$-step ahead prediction for $Y_{n+h}$.
\end{enumerate}
\end{algo}

Several remarks are in order. The proposed algorithm is conceptually simple and allows for a number of immediate extensions and improvements as it is not bound by an assumed FAR structure or, in fact, any other particular functional time series specification. This is important because there is no well developed theory for functional versions of the the well-known linear ARMA time series models ubiquitous in univariate and multivariate settings. Moreover, if an FAR($p$) structure is indeed imposed on $(Y_k\colon k\in\mathbb{Z})$, then it appears plausible that $\bY_1^e,\ldots,\bY_n^e$ should approximately follow a VAR($p$) model. This statement will be made precise in Appendix \ref{s:comp}.

The FAR(1) model should in practice be employed only if it provides a reasonable approximation to the unknown underlying dynamics. To allow for more flexible predictions, higher-order FAR processes could be studied. The proposed methodology offers an automatic way to select the appropriate order $p$ along with the dimensionality $d$ (see Section~\ref{ss:fpe}). It can, in fact, be applied to any stationary functional time series. For example, by utilizing the multivariate innovations algorithm (see Section 11.4 in \cite{bd91}) in the second step of Algorithm \ref{alg:1}. How this is done in the present prediction setting is briefly outlined in Section \ref{s:ex} below. 

It should be emphasized that the numerical implementation of the new prediction methodology is convenient in {\tt R}. For the first step, FPC score matrix $(\bY_1^e\colon\ldots\colon\bY_n^e)$ and corresponding empirical eigenfunctions can be readily obtained with the {\tt fda} package. For the second step, forecasting for the FPC scores can be done in another routine step using the {\tt vars} package in case VAR models are employed. The obtained quantities can be easily combined for obtaining \eqref{fc_kl}.

\section{Predicting functional autoregressions}
\label{s:ffar}

The FAR(1) model \eqref{eq:FAR_mu} is the most often applied functional time series model. It will be used here as a benchmark to compare the proposed methodology to. Without loss of generality it is assumed that $\mu=E[Y_k]=0$. More generally, the higher-order FAR($p$) model
\begin{equation}\label{e:farp}
Y_k=\Psi_1(Y_{k-1})+\cdots+\Psi_p(Y_{k-p})+\varepsilon_k, \qquad k\in\mathbb{Z},
\end{equation}
is considered, assuming throughout that (i) $(\varepsilon_k\colon k\in\mathbb{Z})$ is an i.i.d.\ sequence in $L_H^2$ with $E[\varepsilon_k]=0$, and (ii) the operators $\Psi_j$ are such that equation~\eqref{e:farp} possesses a unique stationary and causal solution. All the above conditions are summarized as Assumption~{\tt FAR}.

\subsection{The standard first-order predictor}\label{ss:ffar}

In order to obtain Bosq's (2000) predictor, estimation of the autoregressive operator $\Psi$ is briefly discussed. The approach is based on a functional version of the Yule-Walker equations. Let then $(Y_k\colon k\in\mathbb{Z})$ be the solution of \eqref{eq:FAR_mu}. Applying $E[\langle \cdot,x\rangle Y_{k-1}]$ to \eqref{eq:FAR_mu} for any $x\in H$, leads to 
\begin{align*}
E[\langle Y_k,x\rangle Y_{k-1}]&=E[\langle \Psi(Y_{k-1}),x\rangle Y_{k-1}]+
E[\langle \varepsilon_k,x\rangle Y_{k-1}]=E[\langle \Psi(Y_{k-1}),x\rangle Y_{k-1}].
\end{align*}
Let again $C(x)=E[\langle Y_1,x\rangle Y_1]$ be the covariance operator of $Y_1$ and also let $D(x)=E[\langle Y_1,x\rangle Y_{0}]$ be the cross-covariance operator of $Y_0$ and $Y_1$. If $\Psi^\prime$ denotes the adjoint operator of $\Psi$, given by the requirement $\langle \Psi(x),y\rangle=\langle x,\Psi^\prime(y)\rangle$, the operator equation $D(x)=C(\Psi^\prime(x))$ is obtained. This formally gives
$
\Psi(x)=D^\prime C^{-1}(x),
$
where $D^\prime(x)=E[\langle Y_0,x\rangle Y_1]$. The operator $D^\prime$ can be estimated by $\hat{D}^\prime(x)=(n-1)^{-1}\sum_{k=2}^n\langle Y_{k-1},x\rangle Y_k$. A more complicated object is the unbounded operator $C^{-1}$. Using the spectral decomposition of $\hat{C}_n$, it can be estimated by
$
\hat{C}^{-1}_n(x)=\sum_{\ell=1}^d\hat{\lambda}_\ell^{-1}\langle \hat{v}_\ell,x\rangle \hat{v}_\ell
$
for an appropriately chosen $d$. Combining these results with an additional smoothing step, using the approximation
$Y_{k} \approx \sum_{\ell=1}^d \langle Y_{k}, \hat{v}_\ell \rangle \hat{v}_\ell$,
gives the estimator
\begin{equation} \label{e:hatPsi}
\tilde \Psi_n (x)
= \frac{1}{n-1}  \sum_{k=2}^{n}  \sum_{\ell=1}^d \sum_{\ell^\prime=1}^d
{\hat \lambda_{\ell}}^{-1}\langle x,  \hat{v}_\ell \rangle \langle Y_{k-1}, \hat{v}_\ell \rangle
\langle Y_{k}, \hat{v}_{\ell^\prime} \rangle \hat{v}_{\ell^\prime}.
\end{equation}
for $\Psi(x)$. This is the estimator of \cite{b00}. It gives rise to the functional predictor
\begin{equation}\label{fc_standard}
\tilde{Y}_{n+1}=\tilde\Psi_n (Y_n)
\end{equation}
for $Y_{n+1}$. Theorem~8.7 of \cite{b00} provides the strong consistency of $\tilde\Psi$ under certain technical assumptions. A recent result of \cite{hki13} (see their Corollary~2.1) shows that consistent predictions (meaning that  $\|\Psi(Y_{n})-\tilde\Psi(Y_n)\|\convP 0$) can be obtained in the present setting if the innovations $(\varepsilon_k\colon k\in\mathbb{Z})$ are elements of $L_H^4$. For these results to hold, it is naturally required that $d=d_n\to \infty$. The choice of $d_n$ crucially depends on the decay rate of the eigenvalues of $C$ as well as on the spectral gaps (distances between eigenvalues). As these parameters are unknown, a practical guideline for the dimension reduction is needed. An approach to this problem in the context of this paper will be provided in Section~\ref{ss:fpe}.

\subsection{Fitting vector autoregressions to FPC scores}\label{ss:comp}

The goal of this section is to show that the one-step predictors $\hat Y_{n+1}$ in \eqref{fc_kl}, based on fitting VAR(1) models in Step 2 of Algorithm \ref{alg:1}, and $\tilde Y_{n+1}$ in \eqref{fc_standard} are asymptotically equivalent for FAR(1) processes. This statement is justified in the next theorem.

\begin{theorem}\label{th:equiv_VAR_FAR}
Suppose model \eqref{eq:FAR_mu} and let Assumption~{\tt FAR} hold. Assume that a VAR(1) model is fit to $\bY_1^e,\ldots,\bY_n^e$ by means of ordinary least squares. The resulting predictor \eqref{fc_kl} is asymptotically equivalent to \eqref{fc_standard}. More specifically, if for both estimators the same dimension $d$ is chosen, then
\[
\|\hat{Y}_{n+1}-\tilde{Y}_{n+1}\|=O_P\bigg(\frac{1}{n}\bigg)\qquad (n\to\infty).
\]
\end{theorem}
The proof of Theorem~\ref{th:equiv_VAR_FAR} is given in Section~\ref{ss:compv_f-b}, where the exact difference between the two predictors is detailed. These computations are based on a more detailed analysis given in Section~\ref{ss:compv_f} which reveals that the FPC score vectors $\bY_1^e,\ldots,\bY_n^e$ follow indeed a VAR(1) model, albeit the non-standard one
\[
\bY_k^e=B_d^e\bY_{k-1}^e+\bdelta_k,
\qquad k=2,\ldots,n,
\]
where the matrix $B_d^e$ is random and the errors $\bdelta_k$ depend on the lag $\bY_{k-1}^e$ (with precise definitions being given in Section~\ref{ss:compv_f}). Given this structure, one might suspect that the use of generalized least squares, GLS, could be advantageous. This is, however, not the case. Simulations not reported in this paper indicate that the gains in efficiency for GLS are negligible in the settings considered. This is arguably due to the fact that possible improvements may be significant only for small sample sizes for which, in turn, estimation errors more than make up the presumed advantage.

Turning to the case of FAR($p$) processes, notice first that Theorem~\ref{th:equiv_VAR_FAR} can be established for the more general autoregessive Hilbertian model (ARH(1)). In this case, the space $L^2([0,1])$ is replaced by a general separable Hilbert space. The proof remains literally unchanged.  Using this fact, a version of Theorem~\ref{th:equiv_VAR_FAR} for higher-order functional autoregressions can be derived by a change of Hilbert space. Following the approach in Section~5.1 of \cite{b00}, write the FAR($p$) process~\eqref{e:farp} in state space form
\begin{equation}\label{e:FARp}
\left(\begin{array}{l}
Y_k \\
Y_{k-1} \\
\phantom{Y}\vdots \\
Y_{k-p+1}
\end{array}
\right)
=
\left(\begin{array}{cccc}
\Psi_1 & \cdots & \Psi_{p-1} & \Psi_p \\
\mathrm{Id} & & & 0 \\
& \ddots & & \vdots \\
& & \mathrm{Id} & 0
\end{array}
\right)
\left(\begin{array}{l}
Y_{k-1} \\
Y_{k-2} \\
\phantom{Y}\vdots \\
Y_{k-p}
\end{array}
\right)
+
\left(\begin{array}{c}
\varepsilon_k \\
0 \\
\vdots \\
0
\end{array}
\right).
\end{equation}
The left-hand side of \eqref{e:FARp} is a $p$-vector of functions. It takes values in the space $H_p=(L^2[0,1])^p$. The matrix on the right-hand side of \eqref{e:FARp} is a matrix of operators which will be denoted by $\Psi^*$. The components $\mathrm{Id}$  and 0 stand for the identity and the zero operator on $H$, respectively.  Equipped with the inner product $\langle x,y\rangle_p=\sum_{j=1}^p\langle x_j,y_j\rangle$ the space $H_p$ defines a Hilbert space. Setting $X_k=(Y_k,\ldots,Y_{k-p+1})'$ and $\delta_k=(\varepsilon_k,0,\ldots,0)'$, equation  \eqref{e:FARp} can be written as $X_k=\Psi^*(X_{k-1})+\delta_k$, with $\delta_k\in L_{H_p}^2$. Now, in analogy to  \eqref{fc_kl} and \eqref{fc_standard}, one can derive the vector-functional predictors $\hat X_k=(\hat X_{k}^{(1)},\ldots,\hat X_k^{(p)})^\prime$ and $\tilde X_k=(\tilde X_{k}^{(1)},\ldots,\tilde X_k^{(p)})^\prime$ and obtain that  $\|\hat X_k-\tilde X_k\|_p=O_P(1/n)$, where $\|x\|_p=\sqrt{\langle x,x\rangle_p}$. Then, the following corollary is immediate.
\begin{corollary}\label{cor:equiv_VAR_FAR}
Consider the FAR($p$) model \eqref{e:farp} and let Assumption~{\tt FAR} hold.  Further suppose that  $\|(\Psi^*)^{k_0}\|_\mathcal{L}<1$ for some $k_0\geq 1$. Then setting 
$\hat Y_k=\hat X_k^{(1)}$ and $\tilde Y_k=\tilde X_k^{(1)}$ one obtains
$
\|\hat{Y}_{n+1}-\tilde{Y}_{n+1}\|=O_P(1/n),
$ as $n\to\infty$.
\end{corollary}

\subsection{Assessing the error caused by dimension reduction}\label{s:3.3new}

Assume the underlying functional time series to be the causal FAR($p$) process. In the population setting, meaning the model is fully known, the best linear one-step ahead prediction (in the sense of mean-squared loss) is $Y^*_{n+1}=\Psi_1(Y_{n})+\cdots \Psi_p(Y_{n-p+1})$, provided $n\geq p$. In this case, the smallest attainable mean-squared prediction error is $\sigma^2:=E[\|\varepsilon_{n+1}\|^2]$. Both estimation methods described in Sections~\ref{ss:ffar} and \ref{ss:comp}, however, give predictions that live on a $d$-dimensional subspace of the original function space. This dimension reduction step clearly introduces a bias, whose magnitude is bounded in this section. It turns out that the bias becomes negligible as $d\to\infty$, thereby providing a theoretical justification for the proposed methodology described in the next subsection.

Unlike in the previous section, it will be avoided to build the proposed procedure on the state space representation~\eqref{e:FARp}. Rather a VAR($p$) model is directly fit by means of ordinary least squares to the $d$-dimensional score sequence. Continuing to work on the population level, the theoretical predictor 
\[
\hat{Y}_{n+1}=\hat y_{n+1,1}v_1+\ldots+\hat y_{n+1,d}v_d,
\]
is analyzed, where $y_{k,\ell}=\langle Y_k,v_\ell\rangle$ and $\hat y_{k,\ell}$ its one-step ahead linear prediction. Recall that a bounded linear operator $A$ is called Hilbert-Schmidt if, for some orthonormal basis $(e_\ell\colon \ell\in\mathbb{N})$, $\|A\|_\mathcal{S}^2=\sum_{\ell=1}^\infty\|A(e_\ell)\|^2<\infty$. Note that $\|\cdot\|_\mathcal{S}$ defines a norm on the space of compact operators which can be shown to be independent of the choice of basis $(e_\ell\colon \ell\in\mathbb{N})$.

\begin{theorem}\label{th:prederror}
Consider the FAR($p$) model \eqref{e:farp} and suppose that Assumption~{\tt FAR} holds. Suppose further that $\Psi_1,\ldots, \Psi_p$ are  Hilbert-Schmidt operators. Then
\begin{equation}\label{e:MSEbound1}
E\big[\|Y_{n+1}-\hat Y_{n+1}\|^2\big]\leq\sigma^2+\gamma_d,
\end{equation}
where
\[
\gamma_d=\bigg(1+\bigg[\sum_{j=1}^p\psi_{j;d}\bigg]^2\bigg)\sum_{\ell=d+1}^\infty\lambda_\ell
\qquad\text{and}\qquad
\psi_{j;d}^2=\sum_{\ell=d+1}^\infty\|\Psi_j(v_\ell)\|^2.
\]
\end{theorem}

The proof of Theorem \ref{th:prederror} is given in Appendix~\ref{proof:th3.2}.

The constant $\gamma_d$ bounds the additional prediction error due to dimension reduction. It decomposes into two terms. The first is given by the fraction of variance explained by the principal components $(v_\ell\colon\ell>d)$. The second term gives the contribution these principal components make to the Hilbert-Schmidt norm of the $\Psi_j$.  Note that $\psi_{j;d}\leq\|\Psi_j\|_{\mathcal{S}}$ and that $\sum_{\ell=1}^\infty\lambda_\ell=\sigma^2$. As a simple consequence, the error in \eqref{e:MSEbound1} tends indeed to $\sigma^2$ for $d\to\infty$.

This useful result, however, does not provide a practical guideline for choosing $d$ in the proposed algorithm because the bound in \eqref{e:MSEbound1} becomes smaller with increasing $d$. Rather $\gamma_d$ has to be viewed as the asymptotic error due to dimension reduction, when $d$ is fixed and $n\to\infty$. In practice one does not have full information on the model for the observations $Y_1,\ldots,Y_n$ and consequently several quantities, such as the autocovariance structure of the score vectors, have to be estimated. Then, with larger $d$, the variance of these estimators increases. In the next section, a novel criterion is provided that allows to simultaneously choose the dimension $d$ and the order $p$ in dependence of the sample size $n$. This is achieved with the objective of minimizing the mean-squared prediction error MSE.

\subsection{Model and dimension selection}
\label{ss:fpe}

Given that the objective of this paper is prediction, it makes sense to choose the model to be fitted to the data as well as the dimension $d$ of the proposed approach such that the MSE is minimized.  Population principal components are still considered (recalling that estimators are $\sqrt{n}$-consistent), but in contrast to the previous section estimated processes are studied. The resulting additional estimation error will now be taken into account.  

Let $(Y_k)$ be a centered functional time series in $L_H^2$. Motivated by Corollary~\ref{cor:equiv_VAR_FAR} VAR($p$) models are fitted to the score vectors. The target is to propose a fully automatic criterion for choosing $d$ and $p$.
By orthogonality of the eigenfunctions $(v_\ell\colon\ell\in\mathbb{N})$ and the fact that the FPC scores $(y_{n,\ell}\colon\ell\in\mathbb{N})$ are uncorrelated, the MSE can be decomposed as
\begin{align*}
E\big[\|Y_{n+1}-\hat Y_{n+1}\|^2\big]
&=E\Bigg[\bigg\|\sum_{\ell=1}^\infty y_{n+1,\ell}v_\ell-\sum_{\ell=1}^d\hat y_{n+1,\ell}v_\ell\bigg\|^2\Bigg]
=E\big[\|\boldsymbol{Y}_{n+1}-\hat{\boldsymbol{Y}}_{n+1}\|^2\big]+\sum_{\ell=d+1}^\infty\lambda_\ell,
\end{align*}
where $\|\cdot\|$ is also used to denote the Euclidean norm of vectors. The process $(\boldsymbol{Y}_{n})$ is again stationary. Assuming that it follows a $d$-variate VAR($p$) model, that is,
$$
\boldsymbol{Y}_{n+1}=\Phi_1\boldsymbol{Y}_n+\cdots +\Phi_p\boldsymbol{Y}_{n-p+1}+\boldsymbol{Z}_{n+1},
$$
with some appropriate white noise $(\boldsymbol{Z}_n)$, it can be shown (see, for example, \cite{l06}) that 
\begin{equation}\label{e:asympnorm}
\sqrt{n}(\boldsymbol{\hat\beta}-\boldsymbol{\beta})\stackrel{d}{\rightarrow} \mathcal{N}_{pd^2}(\boldsymbol{0}, \Sigma_{\boldsymbol{Z}}\otimes\Gamma_p^{-1}),
\end{equation}
where $\boldsymbol{\beta}=\mathrm{vec}([\Phi_1,\ldots,\Phi_p]^\prime)$ and $\boldsymbol{\hat\beta}=\mathrm{vec}([\hat\Phi_1,\ldots,\hat\Phi_p]^\prime)$ is its least squares estimator, and where $\Gamma_p=\mathrm{Var}(\mathrm{vec}[\boldsymbol{Y}_p,\ldots,\boldsymbol{Y}_1])$ and $\Sigma_{\boldsymbol{Z}}=E[\boldsymbol{Z}_1\boldsymbol{Z}_1^\prime]$. Suppose now that the estimator $\boldsymbol{\hat\beta}$ has been obtained from some independent training sample $(\boldsymbol{X}_1,\ldots,\boldsymbol{X}_n)\stackrel{d}{=}(\boldsymbol{Y}_1,\ldots,\boldsymbol{Y}_n)$. Such an assumption is common in the literature. See, for example, the discussion on page 95 of \cite{l06}. It follows then that
\begin{align*}
E \big[ \|\boldsymbol{Y}_{n+1}-\hat{\boldsymbol{Y}}_{n+1}\|^2 \big] &=
E \big[ \|\boldsymbol{Y}_{n+1}-(\hat\Phi_1\boldsymbol{Y}_n+\cdots+\hat\Phi_p\boldsymbol{Y}_{n-p+1})\|^2 \big] \\
&=E\big[\|\boldsymbol{Z}_{n+1}\|^2\big]+E\big[  \|(\Phi_1-\hat\Phi_1)\boldsymbol{Y}_n+\cdots+(\Phi_p-\hat\Phi_p)\boldsymbol{Y}_{n-p+1})\|^2\big] \\
&=\mathrm{tr}(\Sigma_{\boldsymbol{Z}})+E\big[\|[I_p\otimes(\boldsymbol{Y}_n',\ldots,\boldsymbol{Y}_{n-p+1}')](\boldsymbol{\beta}-\boldsymbol{\hat\beta})\|^2\big].
\end{align*}
The independence of $\boldsymbol{\hat\beta}$ and $(\boldsymbol{Y}_1,\ldots,\boldsymbol{Y}_n)$ yields that
\begin{align*}
E\big[\|[I_p\otimes(\boldsymbol{Y}_n',\ldots,\boldsymbol{Y}_{n-p+1}')](\boldsymbol{\beta}-\boldsymbol{\hat\beta})\|^2\big]
&=E\big[\mathrm{tr}\big\{(\boldsymbol{\beta}-\boldsymbol{\hat\beta})'[I_p\otimes\Gamma_p](\boldsymbol{\beta}-\boldsymbol{\hat\beta})\big\}\big]\\
&=
\mathrm{tr}\big\{[I_p\otimes\Gamma_p]E\big[(\boldsymbol{\beta}-\boldsymbol{\hat\beta})(\boldsymbol{\beta}-\boldsymbol{\hat\beta})'\big]\big\}.
\end{align*}
Using \eqref{e:asympnorm}, it follows that the last term is
\[
\frac{1}{n}\left(\mathrm{tr}\left[\Sigma_{\boldsymbol{Z}}\otimes I_{pd}\right]+o(1)\right)\sim\frac{pd}{n}\,\mathrm{tr}(\Sigma_{\boldsymbol{Z}}).
\]
(Here $a_n\sim b_n$ means $a_n/b_n\to 1$.)
Combining the previous estimates and replacing $\mathrm{tr}(\Sigma_{\boldsymbol{Z}})$ by $n(n-pd)^{-1}\mathrm{tr}(\hat \Sigma_{\boldsymbol{Z}})$ , leads to
\[
E\big[\|Y_{n+1}-\hat Y_{n+1}\|^2\big]\approx \frac{n+pd}{n-pd}\,\mathrm{tr}(\hat\Sigma_{\boldsymbol{ Z}})+\sum_{\ell>d}\lambda_\ell.
\]
It is therefore proposed to jointly select the order $p$ and the dimension $d$ as the minimizers of the functional final prediction error-type criterion
\begin{equation}
\label{eq:fFPE}
\mathrm{fFPE}(p,d)=\frac{n+pd}{n-pd}\,\mathrm{tr}(\hat\Sigma_{\boldsymbol{Z}})+\sum_{\ell>d}\hat\lambda_\ell.
\end{equation}
With the use of the functional FPE criterion, the proposed prediction methodology becomes fully data driven and does not need the additional subjective specification of tuning parameters. It is in particular noteworthy that the selection of $d$ is now made in dependence of the sample size $n$. The excellent practical performance of this method is demonstrated in Sections~\ref{s:sim} and~\ref{s:appl}.

It should finally be noted that in a multivariate context \cite{a69} originally suggested the use of the log-determinant in place of the trace in \eqref{eq:fFPE} so as to make his FPE criterion equivalent to the AIC criterion (see \cite{l06}). Here, however, the use of the trace is recommended, since this puts the two terms in \eqref{eq:fFPE} on the same scale. 

\section{Prediction with covariates}\label{s:covariates}

In many practical problems, such as in the particulate matter example presented in Section \ref{s:appl}, predictions could not only contain lagged values of the functional time series of interest, but also other exogenous covariates. These covariates might be scalar, vector-valued and functional. Formally the goal is then to obtain a predictor $\hat Y_{n+h}$ given observations of the curves $Y_1,\ldots,Y_n$ and a number of covariates $X_n^{(1)},\ldots,X_n^{(r)}$. The exogenous variables need not be defined on the same space. For example, $X_n^{(1)}$ could be scalar, $X_n^{(2)}$ a function and $X_n^{(3)}$ could contain lagged values of $X_n^{(2)}$. The following adaptation of the methodology given in Algorithm \ref{alg:1} is derived under the assumption that $(Y_k\colon k\in\mathbb{Z})$ as well as the covariates $(X_n^{(i)}\colon n\in\mathbb{N})$ are stationary processes in their respective spaces. The modified procedure is summarized in Algorithm \ref{alg:3}.
\begin{algo}
\caption{Functional Prediction with Exogenous Covariates}
\label{alg:3}
\begin{enumerate}
\item (a) Fix $d$. For $k=1,\ldots,n$, use the data $Y_1,\ldots,Y_n$ to compute the vectors
\[
\bY_k^e=(y_{k,1}^e,\ldots,y_{k,d}^e)^\prime,
\]
containing the first $d$ empirical FPC scores $y_{k,\ell}^e=\langle Y_k,\hat{v}_\ell\rangle$.

(b) For a functional covariate, fix $d^\prime$. For $k=1,\ldots,n$, use the data $X_1,\ldots,X_n$ to compute the vectors
\[
\bX_k^e=(x_{k,1}^e,\ldots,x_{k,d^\prime}^e)^\prime,
\]
containing the first $d^\prime$ empirical FPC scores $x_{k,\ell}^e=\langle X_k,\hat{w}_\ell\rangle$. Repeat this step for each functional covariate. 

(c) Combine all covariate vectors into one vector $\bR_n^e=(R_{n1}^e,\ldots,R_{nr}^e)^\prime$.
\item Fix $h$. Use $\bY_1^e,\ldots,\bY_n^e$ and $\bR_n^e$ to determine the $h$-step ahead prediction
\[
\hat\bY_{n+h}^e=(\hat{y}^e_{n+h,1},\ldots,\hat{y}^e_{n+h,d})^\prime
\]
for $\bY_{n+h}^e$ with an appropriate multivariate algorithm.
\item Use the functional object
\[
\hat{Y}_{n+h}=\hat{y}_{n+h,1}^e\,\hat{v}_1+\cdots+\hat{y}_{n+h,d}^e\hat{v}_d
\]
as $h$-step ahead prediction for $Y_{n+h}$.
\end{enumerate}
\end{algo}

The first step of Algorithm \ref{alg:3} is expanded compared to Algorithm \ref{alg:1}. Step 1(a) performs FPCA on the response time series curves $Y_1,\ldots,Y_n$. In Step 1(b), all functional covariates are first transformed via FPCA into empirical FPC score vectors. For each functional covariate, a different number of principal components can be selected. Vector-valued and scalar covariates can be used directly. All exogenous covariates are finally combined into one vector $\bR_n^e$ in Step 1(c).

Details for Step~2 and the one-step ahead prediction case $h=1$ could be as follows. Since stationarity is assumed for all involved processes, the resulting FPC scores form stationary time series. Define hence
\[
\Gamma_{\bY\bY}(i)=\cov(\bY_k^e,\bY_{k-i}^e),\qquad
\Gamma_{\bY\bR}(i)=\cov(\bY_k^e,\bR_{k-i}^e),\qquad
\Gamma_{\bR\bR}=\cov(\bR_k^e,\bR_{k}^e)
\]
and notice that these matrices are independent of $k$. Fix $m\in\{1,\ldots,n\}$. The best linear predictor $\hat \bY_{n+1}^e$ of $\bY_{n+1}^e$ given the vector variables $\bY_n^e,\ldots,\bY_{n-m+1}^e,\bR_n^e$ can be obtained by projecting each component $y_{n+1,\ell}^e$ of $\bY_{n+1}^e$ onto
$
\overline{\mathrm{sp}}\{y_{k,i}^e,R_{nj}^e|\, 1\leq i\leq d,\, 1\leq j\leq r,\, n-m+1\leq k\leq n\}.
$
Then there exist $d\times d$ matrices $\Phi_i$ and a $d\times r$ matrix $\Theta$, such that
\[
\hat \bY_{n+1}^e=\Phi_1\bY_n^e+\Phi_2\bY_{n-1}^e+\cdots+\Phi_m\bY_{n-m+1}^e+\Theta\bR_n^e.
\]
Using the projection theorem, it can be easily shown that the matrices $\Phi_1,\ldots,\Phi_m$ and $\Theta$ are characterized by the equations
\begin{align*}
\Gamma_{\bY\bY}(i+1)&=\Phi_1\Gamma_{\bY\bY}(i)+\cdots+\Phi_m\Gamma_{\bY\bY}(i+1-m)+\Theta\Gamma_{\bR\bY}(i),\quad i=0,\ldots, m-1;\\
\Gamma_{\bY\bR}(1)&=\Phi_1\Gamma_{\bY\bR}(0)+\cdots+\Phi_m\Gamma_{\bY\bR}(1-m)+\Theta\Gamma_{\bR\bR}.
\end{align*}
Let
\begin{align*}
\Gamma=\begin{pmatrix}
\Gamma_{\bY\bY}(0)& \Gamma_{\bY\bY}(1) &\cdots& \Gamma_{\bY\bY}(m-1)&\Gamma_{\bY\bR}(0)\\
\Gamma_{\bY\bY}(-1)&\Gamma_{\bY\bY}(0)&\cdots&\Gamma_{\bY\bY}(m-2)&\Gamma_{\bY\bR}(-1)\\
\vdots&\vdots&\ddots&\vdots&\vdots\\
\Gamma_{\bY\bY}(1-m)&\Gamma_{\bY\bY}(2-m)&\cdots&\Gamma_{\bY\bY}(0)&\Gamma_{\bY\bR}(1-m)\\
\Gamma_{\bR\bY}(0)&\Gamma_{\bR\bY}(1)&\cdots&\Gamma_{\bR\bY}(m-1)&\Gamma_{\bR\bR}(0)
\end{pmatrix}.
\end{align*}
Assuming that $\Gamma$ has full rank, it follows that
\begin{align*}
(\Phi_1,\Phi_2,\ldots,\Phi_m,\Theta)&=(\Gamma_{\bY\bY}(1),\ldots,\Gamma_{\bY\bY}(m),\Gamma_{\bY\bR}(1))\Gamma^{-1}.
\end{align*}
The matrices $\Gamma_{\bY\bY}(i)$, $\Gamma_{\bY\bR}(i)$ and  $\Gamma_{\bR\bR}$ have to be replaced in practice by the corresponding sample versions.  This explains why predictions should not be made conditional on all data $\bY_1,\ldots,\bY_n$. It would involve the matrices $\Gamma_{\bY\bY}(n),\Gamma_{\bY\bY}(n-1),\ldots$ which cannot be reasonably estimated from the sample. In the application of Section~\ref{s:appl} a VARX($p$) model of dimension~$d$ is fitted. The dimension~$d$ and the order~$p$ are selected by the adjusted fFPE criterion \eqref{eq:fFPEX}

\section{Additional options}
\label{s:ex}

\subsection{Using the innovations algorithm}

The proposed methodology has been developed with a focus on functional autoregressive processes. For this case a fully automatic prediction procedure has been constructed in Section \ref{ss:fpe}. It should be noted, however, that other options are generally available to the practitioner as well if one seeks to go beyond the FAR framework. One way to do this would be to view the fitted FAR process as a best approximation to the underlying stationary functional time series in the sense of the functional FPE-type criterion in \ref{ss:fpe}. 

In certain cases a more parsimonious modeling could be achieved if one instead used the innovations algorithm in Step 2 of Algorithm \ref{alg:1}. The advantage of the innovations algorithm is that it can be updated quickly when new observations arrive. It should be particularly useful if one has to predict functional moving average processes that have an infinite functional autoregressive representation with coefficient operators whose norms only slowly decay with the lag. The application of Algorithm~\ref{alg:2} requires the estimation of covariances $\Gamma(k)$ for increasing lag $k$. Such estimates are less reliable the smaller $n$ and the larger $k$. Therefore including too many lag values has a negative effect on the estimation accuracy. If estimated eigenfunctions and the covariance matrices $\hat\Gamma(k)$ are replaced by population analogues, then this algorithm gives the best linear prediction (in mean square sense) of the population FPC scores based on the last $m$ observations. 

\begin{algo}
\caption{The Innovations Algorithm for Step 2 in Algorithm \ref{alg:1}}
\label{alg:2}
\begin{enumerate}
\item Fix $m\in \{1,\ldots,n\}$. The last $m$ observations will be used to compute the predictor.
\item For $k=0,1,\ldots,m$, compute
\[ \hat\Gamma(k)=\frac{1}{k}\sum_{j=1}^k(\hat\bY_{j}^e-\bar{\bY}^e)(\hat\bY_{j}^e-\bar{\bY}^e)^\prime,
\]
where $\bar{\bY}^e=\frac{1}{n}\sum_{k=1}^n\hat\bY_{k}^e$.
\item Set
\[
\hat\bY_{n+1}^e=\sum_{j=1}^m\Theta_{mj}(\bY_{n+1-j}^e-\hat\bY_{n+1-j}^e),
\]
where
\begin{align*}
&\Theta_{00}=\hat\Gamma(0),\\
&\Theta_{m,m-k}=\left(\hat\Gamma(n-k)-\sum_{j=0}^{k-1}\Theta_{m,m-j}\Theta_{j0}\Theta^\prime_{k,k-j}\right)\Theta_{k0}^{-1},\qquad k=0,\ldots,m-1,\\
&\Theta_{m0}=\hat\Gamma(0)-\sum_{j=0}^{m-1}\Theta_{m,m-j}\Theta_{j0}\Theta^\prime_{m,m-j}.
\end{align*}
The recursion is solved in the order $\Theta_{00};\Theta_{11},\Theta_{10};\Theta_{22},\Theta_{21},\Theta_{20};\ldots$
\end{enumerate}
\end{algo}

\begin{algo}
\caption{Algorithm for determining prediction bands}
\label{alg:predbands}
\begin{enumerate}
\item Compute the $d$-variate score vectors $\bY_{1}^e,\ldots,\bY_{n}^e$ and the
sample FPCs $\hat v_1,\ldots,\hat v_d$.
\item For $L>0$ fix $k\in\{L+1,\ldots,n-1\}$ and compute
$$
\hat Y_{k+1}=\hat{y}_{k+1,1}^e\,\hat{v}_1+\cdots+\hat{y}_{k+1,d}^e\hat{v}_d,
$$
where $\hat{y}_{k+1,1}^e,\ldots,\hat{y}_{k+1,d}^e$, are the components of the one-step ahead prediction obtained from $\bY_{1}^e,\ldots,\bY_{k}^e$ by means of a multivariate algorithm.
\item Let $M=n-L$. For $k\in\{1,\ldots,M\}$, define the residuals $\hat\epsilon_{k}=Y_{k+L}-\hat Y_{k+L}$.
\item For $t\in[0,1]$, define $\gamma(t)=\mathrm{sd}(\hat\epsilon_k(t)\colon k=1\ldots,M)$.
\item Determine $\underline{\xi}_\alpha,\overline{\xi}_\alpha$ such that $\alpha\times 100\%$ of the residuals satisfy 
$$ 
-\underline{\xi}_\alpha\gamma(t)\leq \hat\epsilon_i(t)\leq \overline{\xi}_\alpha\gamma(t)\quad\text{for all $t\in[0,1]$}.
$$
\end{enumerate}
\end{algo}

\subsection{Prediction bands}

To assess the forecast accuracy, a method for computing uniform prediction bands is provided in this section. The target is to find parameters $\underline{\xi}_\alpha,\overline{\xi}_\alpha\geq 0$, such that, for a given $\alpha\in (0,1)$ and $\gamma\colon[0,1]\to [0,\infty)$, 
$$
P\Big(\hat Y_{n+1}(t)-\underline{\xi}_\alpha\gamma(t)\leq Y_{n+1}(t)\leq \hat Y_{n+1}(t)+\overline{\xi}_\alpha\gamma(t)\quad\text{for all $t\in[0,1]$}\Big)=\alpha.
$$
There is no a priori restriction on the function $\gamma$, but clearly it should account for the structure and variation of the data. Although this problem is very interesting from a theoretical standpoint, only a practical approach for the determination of $\underline{\xi}_\alpha,\overline{\xi}_\alpha$ and $\gamma$ is proposed here. It is outlined in Algorithm~\ref{alg:predbands}.

The purpose of the parameter $L$ is to ensure a reasonable sample size for the predictions in Step~2 of Algorithm~\ref{alg:predbands}. The residuals $\hat\epsilon_1,\ldots\hat\epsilon_M$ are then expected to be approximately stationary and, by a law of large numbers effect, to satisfy
\begin{align*}
&\frac{1}{M}\sum_{k=1}^MI\Big(-\underline{\xi}_\alpha\gamma(t)\leq \hat\epsilon_k(t)\leq \overline{\xi}_\alpha\gamma(t)\quad\text{for all } t\in[0,1]\Big)\\
&\quad\approx 
P\Big(-\underline{\xi}_\alpha\gamma(t)\leq Y_{n+1}(t)-\hat Y_{n+1}(t)\leq 
\overline{\xi}_\alpha\gamma(t)\quad\text{for all } t\in[0,1]\Big).
\end{align*}
Note that, in Step~1, the principal components $\hat v_1,\ldots,\hat v_d$ have been obtained from the entire sample $Y_1,\ldots,Y_n$ and not just from the first
$k$ observations. The choice of $\gamma$ in Step~4 clearly accounts for the variation of the data. For an intraday time exhibiting a higher volatility there should also be a broader prediction interval. Typically the constants $\underline{\xi}_\alpha$ and $\overline{\xi}_\alpha$ are chosen equal, but there may be situations when this is not desired.

One advantage of this method is that it does not require particular model assumptions. If two competing prediction methods exist, then the one which is performing better on the sample will lead to narrower prediction bands. Simulation results not reported in this paper indicate that Algorithm~\ref{alg:predbands} performs well in finite samples even for moderate sample sizes.

\section{Simulations}
\label{s:sim}

\subsection{General setting}
\label{se:simsetting}
To analyze the finite sample properties of the new prediction method, a comparative simulation study was conducted. The proposed method was tested on a number of functional time series, namely first- and second-order FAR processes, first-order FMA processes and FARMA processes of order (1,2). In each simulation run, $n=200$ (or $1000$) observations were generated of which the first $m=180$ (or $900$) were used for parameter estimation as well as order and dimension selection with the fFPE($p,d$) criterion \eqref{eq:fFPE}. On the remaining $20$ (or $100$) observations one-step ahead predictions and the corresponding squared prediction errors were computed. From these mean (MSE), median (medSE) and standard deviation (SD) were calculated. If not otherwise mentioned, this procedure was repeated $N=100$ times. More details and a summary of the results are given in Sections~\ref{se:scalar}--\ref{se:farma}.

Since in simulations one can only work in finite dimensions, the setting consisted of $D$ Fourier basis functions $v_1,\ldots,v_{D}$ on the unit interval $[0,1]$, which together determine the (finite-dimensional) space  $H=\mathrm{sp}\{v_1,\ldots,v_{D}\}$. Note that an arbitrary element $x\in H$ has the representation $x(t)=\sum_{\ell=1}^{D}c_\ell v_{\ell}(t)$ with coefficients $\mathbf{c}=(c_1,\ldots,c_{D})^\prime$. If $\Psi\colon H\to H$ is a linear operator, then
\begin{align*}
\Psi(x)&=\sum_{\ell=1}^{D}c_\ell\Psi(v_{\ell})
=\sum_{\ell=1}^{D}\sum_{\ell^\prime=1}^{D}c_\ell\langle \Psi(v_\ell),v_{\ell^\prime}\rangle v_{\ell^\prime}
=(\boldsymbol{\Psi} \mathbf{c})^\prime \mathbf{v},
\end{align*}
where $\boldsymbol{\Psi}$ is the matrix whose $\ell$-th column and $\ell^\prime$-th row is $\langle \Psi(v_\ell),v_{\ell^\prime}\rangle$, and $\mathbf{v}=(v_1,\ldots,v_{D})^\prime$ is the vector of basis functions. 
The linear operators needed to simulate the functional time series of interest can thus be represented by a $D\times D$ matrix that acts on the coefficients in the basis function representation of the curves. The corresponding innovations were generated according to
\begin{align}
\varepsilon_k(t)= \sum_{\ell=1}^{D} A_{k,\ell} v_{\ell}(t)
\label{eq:Xrepr},
\end{align}
where $A_{k,\ell}$ are i.i.d.\ normal random variables with mean zero and standard deviations $\sigma_{\ell}$ that will be specified below.

\subsection{Comparison with scalar prediction}
\label{se:scalar}

As mentioned in the introduction, a special case of the proposed method was considered by \cite{hu:07} and \cite{hs:09}. Motivated by the fact that PCA score vectors have uncorrelated components, these authors have proposed to predict the scores individually as univariate time series. This will be referred to as the {\em scalar method}, in contrast to the {\em vector method}\/ promoted in this paper. The scalar method is fast and works well as long as the cross-spectra related to the score vectors are close to zero. However, in general the score vectors have non-diagonal autocorrelations. Then, scalar models are not theoretically justified. To explore the effect of neglecting cross-sectional dependence, FAR(1) time series of length $n=200$ were generated as described above.
For the purpose of demonstration $D=3$ and $\sigma_1=\sigma_2=\sigma_3=1$ were chosen. Two autocovariance operators $\Psi^{(1)}$ and $\Psi^{(2)}$ with corresponding matrices
$$
\boldsymbol{\Psi}^{(1)}=
\begin{pmatrix}
-0.05 & -0.23 & \,\,\,\,\,0.76 \\ 
 \,\,\,\,\,0.80 & -0.05 & \,\,\,\,\,0.04 \\ 
 \,\,\,\,\,0.04 & \,\,\,\,\,0.76 & \,\,\,\,\,0.23 
\end{pmatrix}
\qquad\text{and}\qquad
\boldsymbol{\Psi}^{(2)}=0.8
\begin{pmatrix}
1 & \,\,\,\,\,0 & \,\,\,\,\,0 \\ 
 0 &\,\,\,\,\, 1 &\,\,\,\,\, 0 \\ 
 0 &\,\,\,\,\, 0 & \,\,\,\,\,1 
\end{pmatrix},
$$ were tested. Both matrices are orthogonal with norm $0.8$. In these simple settings it is easy to compute the population autocorrelation function (ACF) of the 3-dimensional FPCA score vectors. The ACF related to the score sequences of the process generated by $\Psi^{(1)}$ is displayed in Figure~\ref{fig:ACF}. It shows that two scores are uncorrelated at lag zero and that there is almost no temporal correlation in the individual score sequences. However, at lags greater than $1$ there is considerable dependence in the cross-correlations between the first and the third score sequence.
\begin{figure}
\begin{center}
\includegraphics[width=10cm]{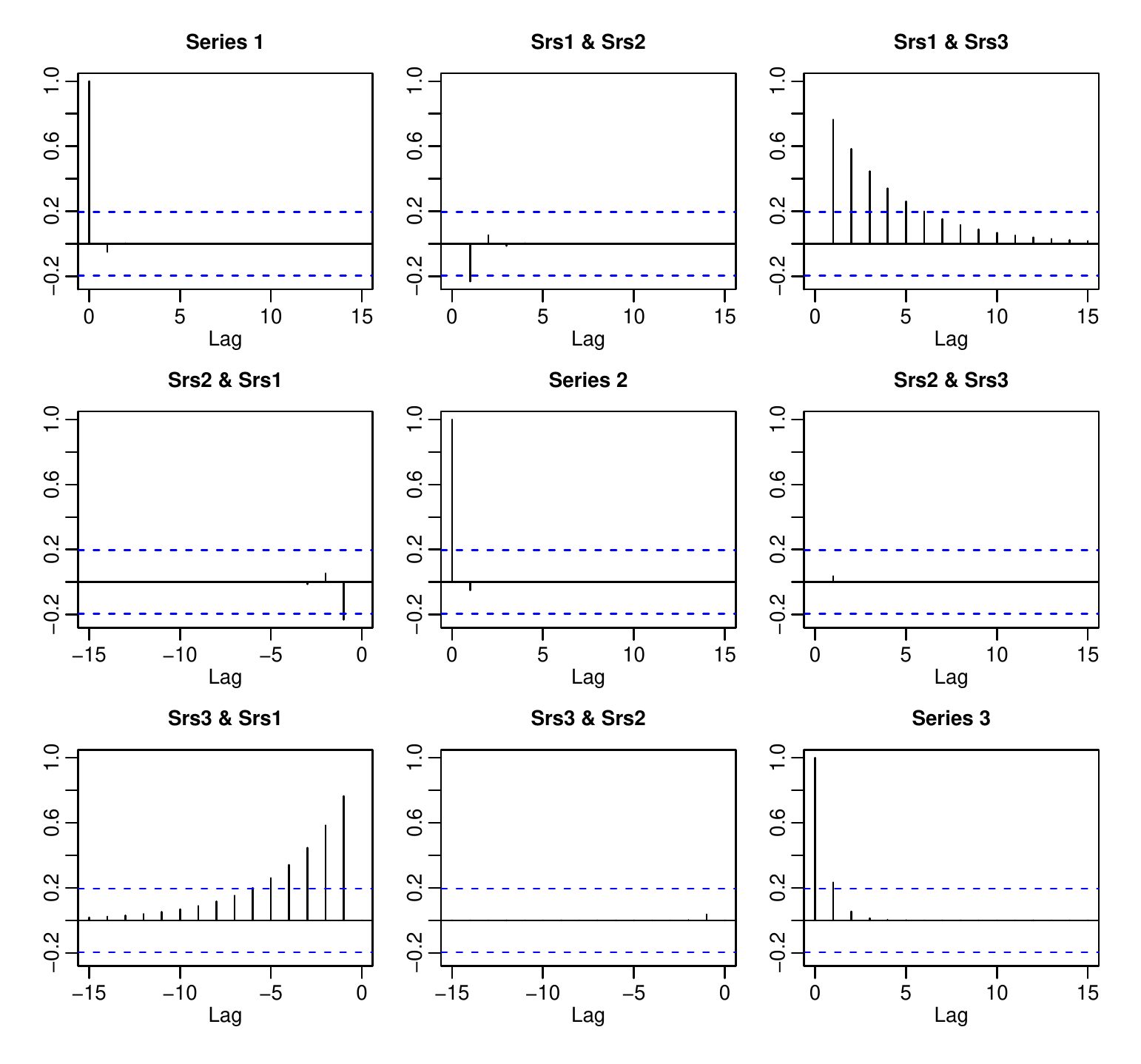}\hfill
\caption{Autocorrelation function for the scores related to the sequence generated from operator~$\Psi^{(1)}$.}
\label{fig:ACF}
\end{center}
\end{figure}
The analogous plot for $\Psi^{(2)}$ would reveal a contrary behavior: while the autocorrelations of the individual score sequences decay slowly, cross-correlations are zero at all lags.

Given these observations, it is expected that the scalar method will do very well in forecasting the scores when data are generated by operator $\Psi^{(2)}$, while it should be not competitive with the vector method if $\Psi^{(1)}$ is used. This conjecture is confirmed in Figure~\ref{fig:ratios} which shows histograms of the ratios 
\begin{equation}\label{eq:ri}
r_i=\frac{\text{MSE vector method}}{\text{MSE scalar method}},\qquad i=1,\ldots,1000,
\end{equation}
obtained from 1000 simulation runs. The grey histogram refers to the time series generated by $\Psi^{(2)}$. It indicates that the scalar method is a bit favorable, as the ratios tend to be slightly larger than one. Contrary to this, a clear superiority of the vector method can be seen when data stem from the sequence generated by $\Psi^{(1)}$. In a majority of the cases, the MSE resulting from the vector methods is less than half as large as the corresponding MSE obtained by the scalar method.  It should also be mentioned that $p$ and $d$ where estimated for the proposed method, while they were fixed at the true values $p=1$ and $d=3$ for the scalar predictions.

\begin{figure}
\begin{center}
\includegraphics[width=.5\textwidth]{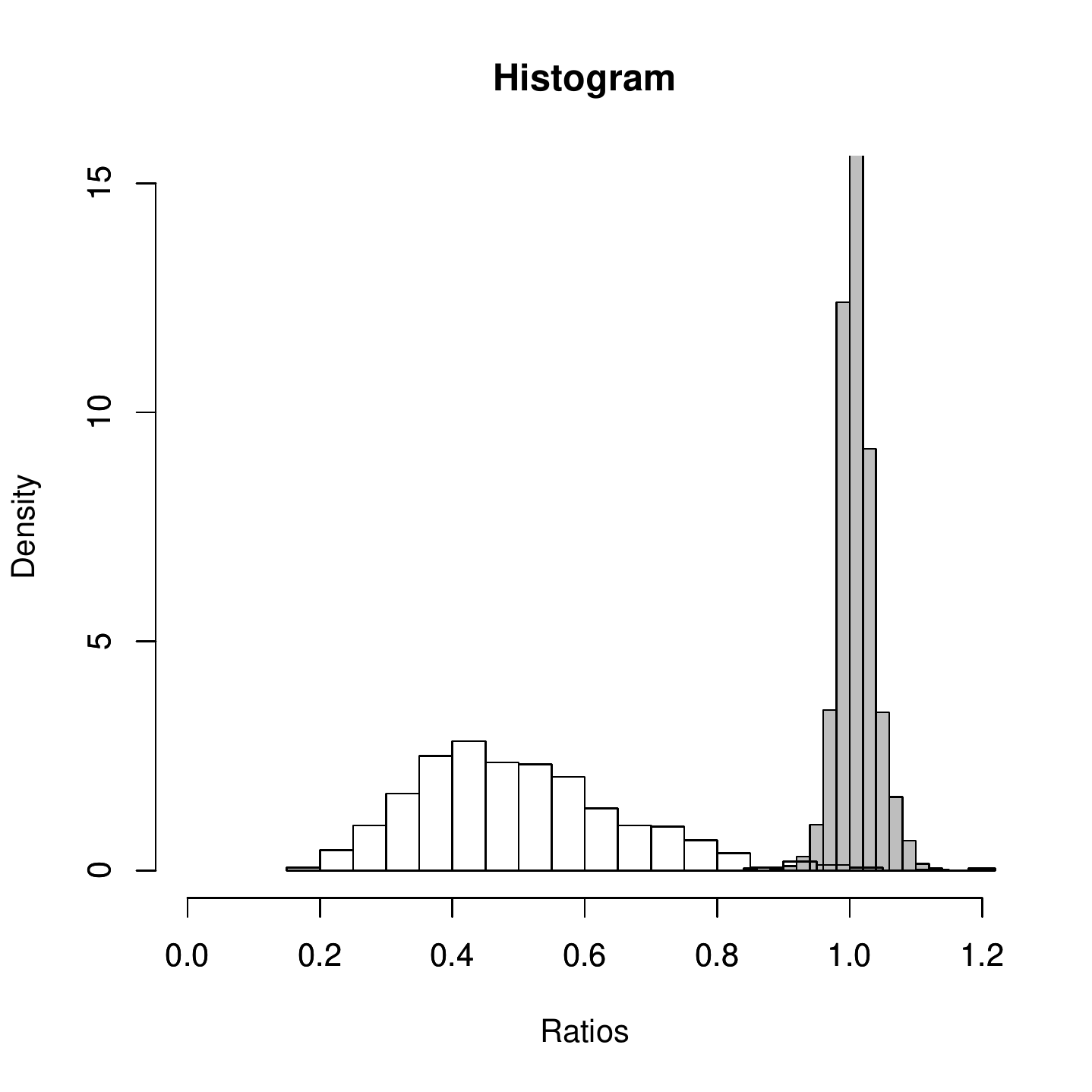}\hfill
\caption{Histogram of the ratios $r_i$ in \eqref{eq:ri} for the FAR(1) processes given by the operators $\Psi^{(1)}$ (white) and $\Psi^{(2)}$ (grey).}
\label{fig:ratios}
\end{center}
\end{figure}

\subsection{Comparison with standard functional prediction}
\label{se:results}

In this section the proposed prediction is compared on FAR(2) processes $Y_k=\Psi_1Y_{k-1}+\Psi_2Y_{k-2}+\varepsilon_k$ to the standard predicton of \cite{b00}. For the latter, the multiple testing procedure of \cite{kr13} was utilized to determine the order $p$ of the FAR model to be fitted. Following these authors, $d$ was chosen as the smallest integer such that the first $d$ principal components explain at least $80\%$ of the variance of the data. To ensure that the multiple testing procedure keeps an overall asymptotic level of 10\%, the levels in three subtests (so testing up to a maximal order $p=3$)  were chosen to be 5\%, 3\% and 2\%, respectively. For ease of reference this method will be referred to as the BKR method. Owing to the results of Section~\ref{s:ffar}, both methods are expected to yield similar results if the order $p$ was known and if the same dimension $d$ was chosen for the two predictors. 

The operators were generated such that $\Psi_1=\kappa_1\,\Psi$ and $\Psi_2=\kappa_2\,\Psi$ with $|\kappa_1|+|\kappa_2|<1$ to ensure stationarity. The case $\kappa_2=0$ yields the FAR(1) process. The operator $\Psi$ was chosen at random. More precisely, choosing $D=21$, a $D\times D$ matrix of independent, zero-mean normal random variables with corresponding standard deviations $\sigma_{\ell\ell^\prime}$ was generated. This matrix was then scaled so that the resulting matrix $\boldsymbol{\Psi}$ has induced norm equal to 1. In every iteration of the simulation runs $\boldsymbol{\Psi}$ was newly generated. Two types of standard deviations for the innovations in \eqref{eq:Xrepr} were chosen, namely
\[
\text{($\sigma$1)}\quad \boldsymbol{\sigma}_1^\prime=(\ell^{-1}\colon \ell=1,\ldots,D)
\qquad\text{and}\qquad
\text{($\sigma$2)}\quad \boldsymbol{\sigma}_2^\prime=(1.2^{-\ell}\colon \ell=1,\ldots,D).
\]
Note that if $\Psi\colon L^2\to L^2$, then $\langle \Psi(v_\ell),v_{\ell^\prime}\rangle\to 0$ if $\ell\to\infty$ or $\ell'\to\infty$ by the Riemann-Lebesgue lemma. This will be reflected in the corresponding matrices by choosing $\sigma_{\ell\ell^\prime}$ as a decaying sequence in $\ell$ and $\ell^\prime$. In particular we have chosen $((\sigma_{\ell\ell^\prime}))=\boldsymbol{\sigma}_1\boldsymbol{\sigma}_1^\prime$ for setting ($\sigma$1) and $((\sigma_{\ell\ell^\prime}))=\boldsymbol{\sigma}_2\boldsymbol{\sigma}_2^\prime$ for setting ($\sigma$2).

\begin{table}
\begin{center}
\begin{tabular}{ll| lllll | lllll}
\toprule
\multicolumn{2}{l}{}&
\multicolumn{5}{c|}{($\sigma$1)}&\multicolumn{5}{c}{($\sigma$2)} \\
\cmidrule(r){3-12}
$\kappa_1$ & $\kappa_2$ & \text{fFPE} & $\text{MSE}_a$ & $\text{MSE}_b$ & $\text{PVE}_a$ & $\text{PVE}_b$ & \text{fFPE} & $\text{MSE}_a$ & $\text{MSE}_b$ & $\text{PVE}_a$ & $\text{PVE}_b$\\
\cmidrule(r){1-12}
$0.2$ & $0.0$ & 2.29 & 2.32 & 2.31 & 0.40 & 0.83 & 1.61 & 1.59 & 1.58 & 0.71 & 0.83\\
 &  & 2.29 & 2.31 & 2.31 & 0.72 & 0.84 & 1.61 & 1.59 & 1.59 & 0.81 & 0.82\\
\cmidrule(r){1-12}
0.8 & 0.0 & 2.37 & 2.37 & 2.47 & 0.90 & 0.83 & 1.64 & 1.71 & 1.81 & 0.89 & 0.85\\
 & & 2.30 & 2.29 & 2.37 & 0.97 & 0.83 & 1.61 & 1.62 & 1.70 & 0.95 & 0.85\\
\cmidrule(r){1-12}
0.4 & 0.4 & 2.38 & 2.40 & 2.43 & 0.73 & 0.83 & 1.67 & 1.65 & 1.69 & 0.84 & 0.84\\
 & & 2.31 & 2.33 & 2.36 & 0.92 & 0.83 & 1.63 & 1.64 & 1.71 & 0.90 & 0.84\\
\cmidrule(r){1-12}
0.0 & 0.8 & 2.42 & 2.48 & 2.94 & 0.83 & 0.83 & 1.66 & 1.72 & 2.28 & 0.87 & 0.85\\
 & & 2.32 & 2.34 & 2.94 & 0.95 & 0.83 & 1.63 & 1.62 & 2.27 & 0.93 & 0.86\\
\bottomrule
\end{tabular}
\end{center}\label{tab:results}
\caption{Functional final prediction error (fFPE), mean squared prediction error based on the fFPE criterion ($\text{MSE}_a$), mean squared prediction error based on BKR ($\text{MSE}_b$), and the corresponding proportions of variance explained by the chosen number of FPCs ($\text{PVE}_a$, $\text{PVE}_b$). The first row in each setting $(\kappa_1,\kappa_2)$ corresponds to $n=200$, the second row to $n=1000$.}
\end{table}

Results for four pairs of values $(\kappa_1, \kappa_2)$ are shown in Table~\ref{tab:results}. The numbers are averages from 100 iterations of the simulation setting explained in Section~\ref{se:simsetting}. Recall that $10\%$ of the data was used in each simulation run to compute out-of-sample predictions. This means that the MSE's are based on 2,000 forecasts when $n=200$ and 10,000 forecasts when $n=1,000$. The quantity $\text{MSE}_a$ refers to the MSE produced by the proposed method and $\text{MSE}_b$ to the MSE obtained from the BKR method.  Similarly, $\text{PVE}_a$ and $\text{PVE}_b$ give the respective averages of the proportions of variance explained by $d$ principal components, where $d$ is the chosen dimension of the predictor. In summary, the following was found:
\begin{itemize}
\item The proposed approach had slight advantages over BKR in almost all considered settings. For $\kappa_1=0$ and $\kappa_2=0.8$, the BKR method almost always failed to choose the correct order $p$ (see Table~\ref{tab:results2}). In this case $\text{MSE}_b$ was about 30\%--40\% larger than $\text{MSE}_a$. 
\item With increasing sample size $\text{MSE}_a$ decreases and approaches the value of the fFPE. The latter is an estimate for the minimal possible MSE. Contrary to the BKR method, the dimension parameter $d$ chosen by fFPE grows with increasing sample size. This is visualized in Figure~\ref{fig:dimensions}.
\item When both methods choose the correct order $p$, $\text{MSE}_a$ still had a tendency to be smaller than $\text{MSE}_b$. This may arguably be due to the fact that a data driven criterion was applied to optimally select the dimension parameter $d$. It can also be seen that the mean squared prediction errors are relatively robust with respect to the choice of $d$ but quite sensitive to the choice of $p$. In particular, underestimating $p$ can lead to a non-negligible increase of MSE.
\item We have also experimented with $D=51$. The conclusions remain very similar.
\end{itemize}

\begin{figure}
\begin{center}
\includegraphics[width=.5\textwidth]{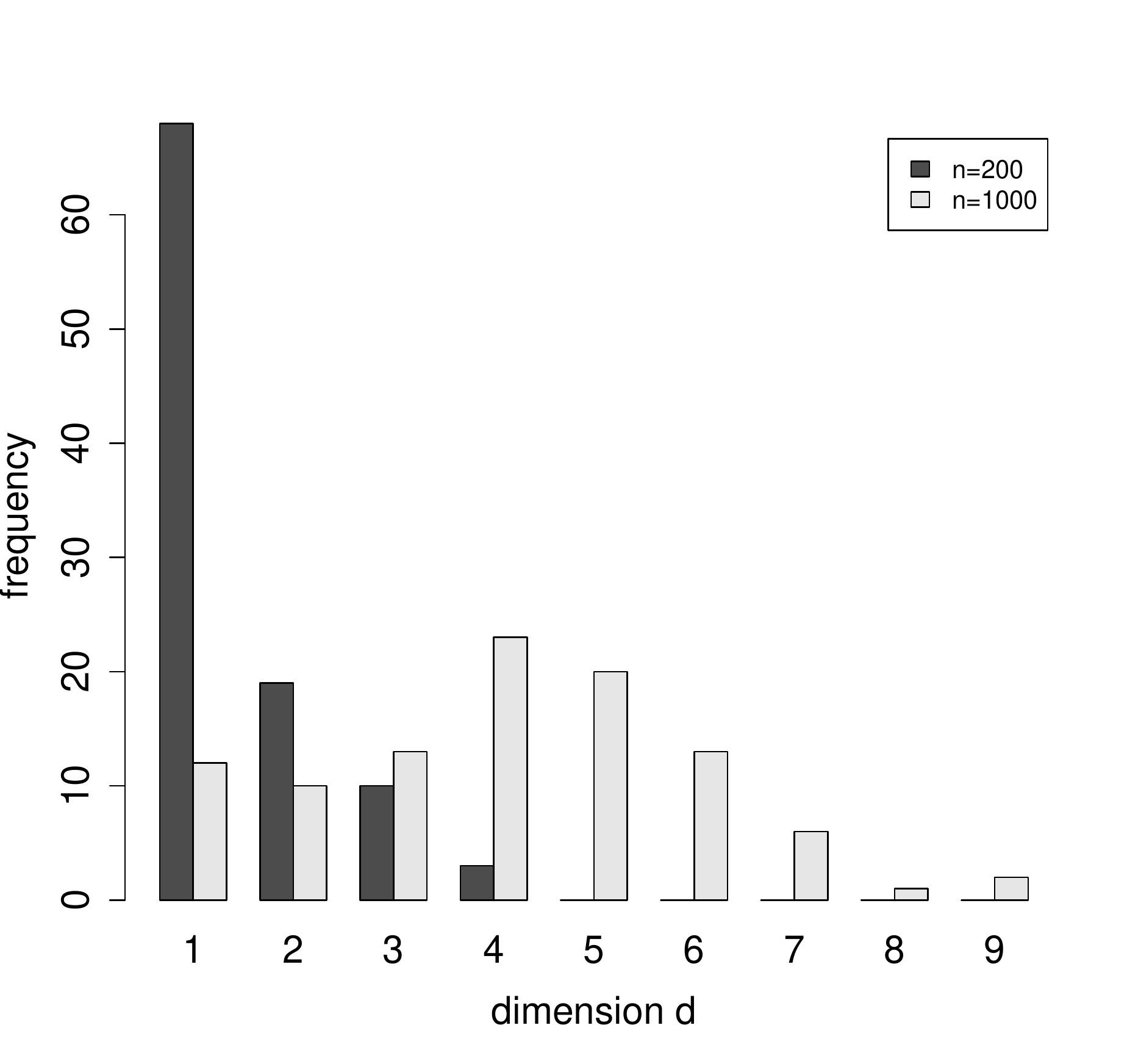}
\caption{Frequencies of the dimensions $d$ chosen by fFPE in 100 simulation runs under setting ($\sigma 1$) and $(\kappa_1,\kappa_2)=(0.2,0.0)$.}
\end{center}
\label{fig:dimensions}
\end{figure}

\begin{table}
\begin{center}
\begin{tabular}{ll| rrrr | rrrr}
\toprule
\multicolumn{2}{l}{}&
\multicolumn{4}{c|}{$n=200$}&\multicolumn{4}{c}{$n=1000$} \\
\cmidrule(r){3-10}
$\kappa_1$ & $\kappa_2$ & $p=0$ & $p=1$ & $p=2$ & $p=3$ & $p=0$ & $p=1$ & $p=2$ & $p=3$\\
\cmidrule(r){1-10}
$0.2$ & $0.0$ & 40 & {\bf 48} & 8 & 4 & 2 & {\bf 94} & 3 & 1\\
 &  & 48 & {\bf 51} & 1 & 0 & 0 & {\bf 98} & 2 & 0\\
\cmidrule(r){1-10}
$0.8$ & $0.0$ & 0 & {\bf 97} & 3 & 0 & 0 & {\bf 100} & 0 & 0\\
 &  & 0 & {\bf 95} & 5 & 0 & 0 & {\bf 81} & 17 & 2\\
\cmidrule(r){1-10}
$0.4$ & $0.4$ & 1 & 3 & {\bf 90} & 6 & 0 & 0 & {\bf 99} & 1\\
 &  & 3 & 3 & {\bf 94} & 0 & 0 & 0 & {\bf 95} & 5\\
\cmidrule(r){1-10}
$0.0$ & $0.8$ & 0 & 0 & {\bf 95} & 5 & 0 & 0 & {\bf 99} & 1\\
 &  & 94 & 0 & {\bf 5} & 1 & 93 & 0 & {\bf 7} & 0\\
\bottomrule
\end{tabular}
\end{center}\label{tab:results2}
\caption{Selected order for different choices of $\kappa_1$ and $\kappa_2$ from 100 iterations under setting ($\sigma$1). For each choice the top (bottom) row represents the order obtained via fFPE (BKR). The number of correctly selected orders is shown in bold.}
\end{table}

\subsection{Beyond functional autoregressions}
\label{se:farma}

To test the proposed procedure also for non-autoregressive functional time series, it was applied to the functional FMA(2) and FARMA(1,2) processes respectively given by the equations
\begin{align}
Y_k&=\varepsilon_k+\Theta\varepsilon_{k-2}, 
\label{fma} \\
Y_k&=\Psi_1Y_{k-1}+\varepsilon_k+\Theta_1\varepsilon_{k-1}+\Theta_2\varepsilon_{k-2},
\label{farma}
\end{align}
with operators $\Theta=.8\Psi$, $\Psi_1=.1\Psi$, $\Theta_1=.1\Psi$ and $\Theta_2=.9\Psi$ randomly generated as above. Both the fFPE-based proposed procedure and the BKR method were applied to time series of length $n=1000$. Since a fitting of long autoregressions is expected the maximal order was set to be 10. The rejection levels for the individual tests of the BKR method were set to achieve an overall level of approximately 10\%. The simulation results are displayed in Table 6.3.
The conclusions of the previous section still hold true. In particular, MSE reductions of 15\%--25\% are seen, with the reduction being slightly greater for the FARMA(1,2) process. The proposed method approximates the given time series structure generally with longer FAR processes with average orders (taken over 100 simulation runs) between $p=4$ and $p=5$ in all four cases. On the other hand, the BKR method largely fails to make adjustments and selects $p=0$ more than 90\% of the time.

\begin{table}
\begin{center}
\begin{tabular}{l| lllll | lllll}
\toprule
\multicolumn{1}{l}{}&
\multicolumn{5}{c|}{($\sigma$1)}&\multicolumn{5}{c}{($\sigma$2)} \\
\cmidrule(r){3-11}
 & \text{fFPE} & $\text{MSE}_a$ & $\text{MSE}_b$ & $\text{PVE}_a$ & $\text{PVE}_b$ & \text{fFPE} & $\text{MSE}_a$ & $\text{MSE}_b$ & $\text{PVE}_a$ & $\text{PVE}_b$\\
\cmidrule(r){1-11}
FMA(2) & 2.37 & 2.39 & 2.80 & 0.92 & 0.83 & 1.65 & 1.64 & 2.12 & 0.90 & 0.85 \\
FARMA(1,2) & 2.38 & 2.42 & 2.96 & 0.82 & 0.83 & 1.65 & 1.67 & 2.24 & 0.90 & 0.85 \\
\bottomrule
\end{tabular}
\end{center}\label{tab:results3}
\caption{As in Table 6.1, but for the functional time series in \eqref{fma} and \eqref{farma} for $n=1000$.}
\end{table}

\section{Predicting particulate matter concentrations}
\label{s:appl}

In order to demonstrate its practical usefulness, the new methodology has been applied to environmental data on pollution concentrations. The observations are half-hourly measurements of the concentration (measured in $\mu gm^{-3}$) of particulate matter with an aerodynamic diameter of less than $10\mu m$, abbreviated PM10, in ambient air taken in Graz, Austria from October 1, 2010 until March 31, 2011. Since epidemiological and toxicological studies have pointed to negative health effects, European Union (EU) regulation sets pollution standards for the level of the concentration. Policy makers have to ensure compliance with these EU rules and need reliable statistical tools to determine, and justify to the public, appropriate measures such as partial traffic regulation (see \cite{shp08}). Accurate predictions are therefore paramount for well informed decision making.

Functional data were obtained as follows. In a first step, very few missing intra-day data points were replaced through linear interpolation. A square-root transformation was then applied to the data in order to stabilize the variance. A visual inspection of the data revealed several extreme outliers around New Years Eve known to be caused by firework activities. The corresponding week was removed from the sample. The data was then centered and adjusted for weekly seasonality by subtracting from each observation the corresponding weekday average. This is done because PM10 concentration levels are significantly different for the weekends when traffic volume is much lower. 
In the next step, 48 observations for a given day were combined into vectors and transformed into functional data using ten cubic $B$-spline basis functions and least squares fitting. The {\tt fda} package available for the statistical software {\tt R} was applied here. Eventually, 175 daily functional observations, say, $Y_1,\ldots,Y_{175}$, were obtained, roughly representing one winter season for which pollution levels are known to be high. They are displayed in the upper left panel of Figure \ref{fig:fd}. Shown in this figure are also the effect of the first three FPCs on the mean curve. Following \cite{rs05}, a multiple (using the factor .5) of the $\ell$th empirical eigenfunction $\hat v_\ell$ was added to and subtracted from the overall estimated mean curve $\hat \mu$ to study the effect of large (small) first, second or third FPC score. Notice that 
\[
Y_k\approx\hat\mu+y_{k1}^e\hat v_1+y_{k2}^e\hat v_2+y_{k3}^e\hat v_3,
\qquad k=1,\ldots,175,
\]
where $y_{k\ell}^e=\langle Y_k,\hat v_\ell\rangle$ are the empirical FPC scores. These combine to explain about 89\% of variability in the data. The upper right panel of Figure \tb{\ref{fig:fd}} indicates that if the first FPC score $y_{k1}^e$, which explains about 72\% of the variation, is large (small), then a positive (negative) shift of the mean occurs. The second and third FPCs are contrasts, explaining respectively 10\% and 7\% of variation, with the second FPC describing an intraday trend and the third FPC indicating whether the diurnal peaks are more or less pronounced (see the lower panel of Figure~\tb{\ref{fig:fd}}).

\begin{figure}
\begin{center}
\includegraphics[width=7cm]{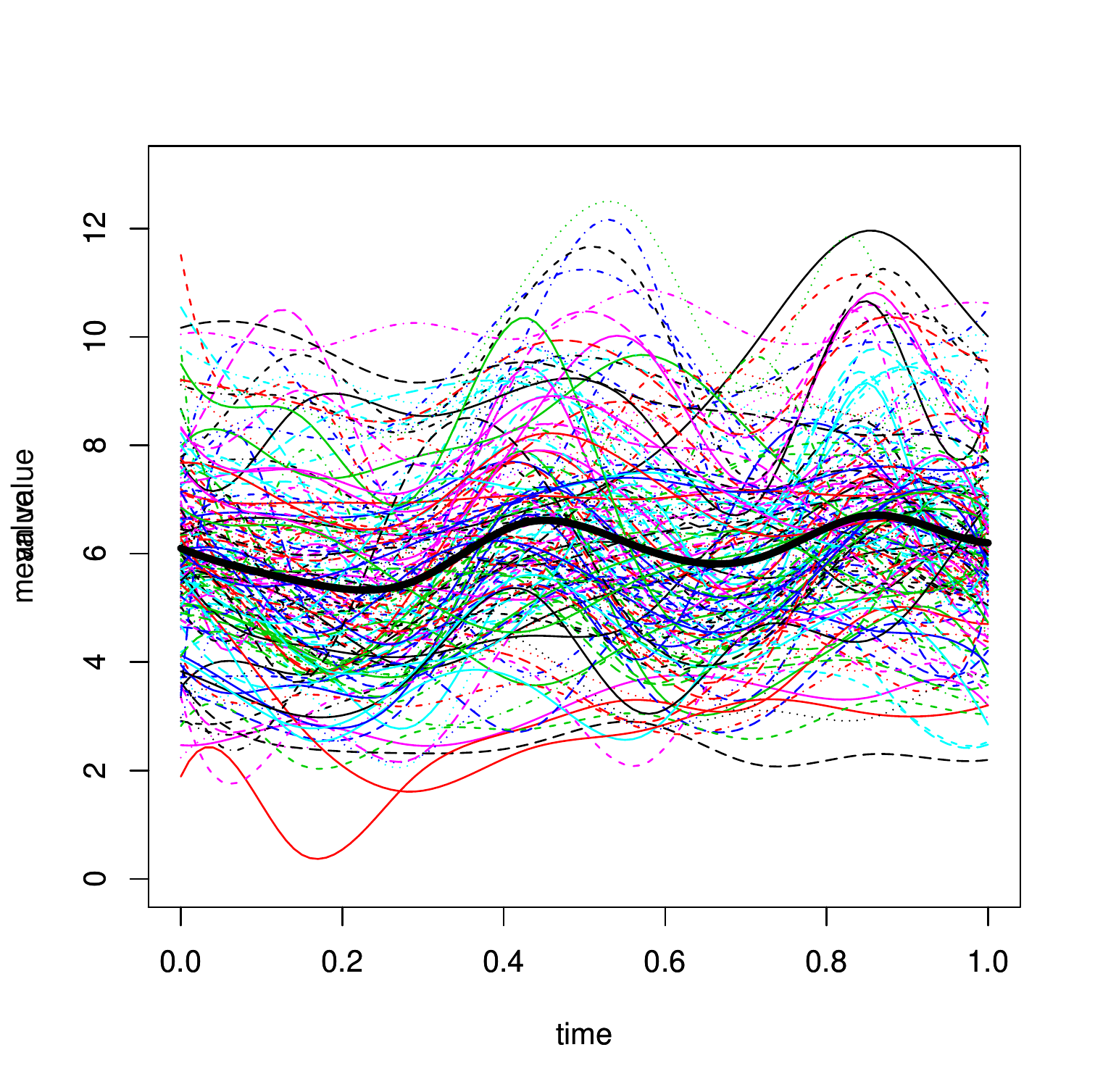} \quad
\includegraphics[width=7cm]{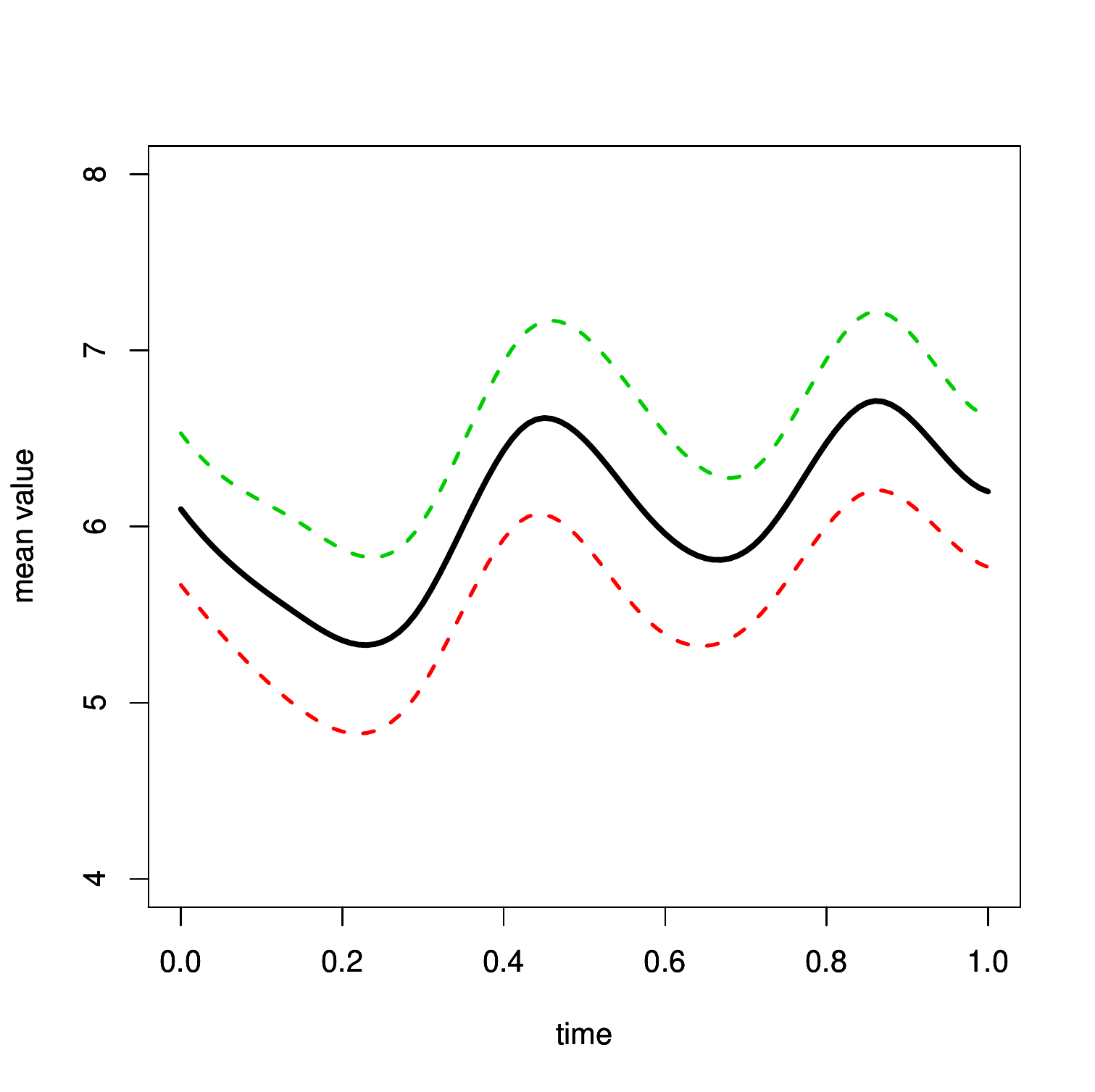}\\
\includegraphics[width=7cm]{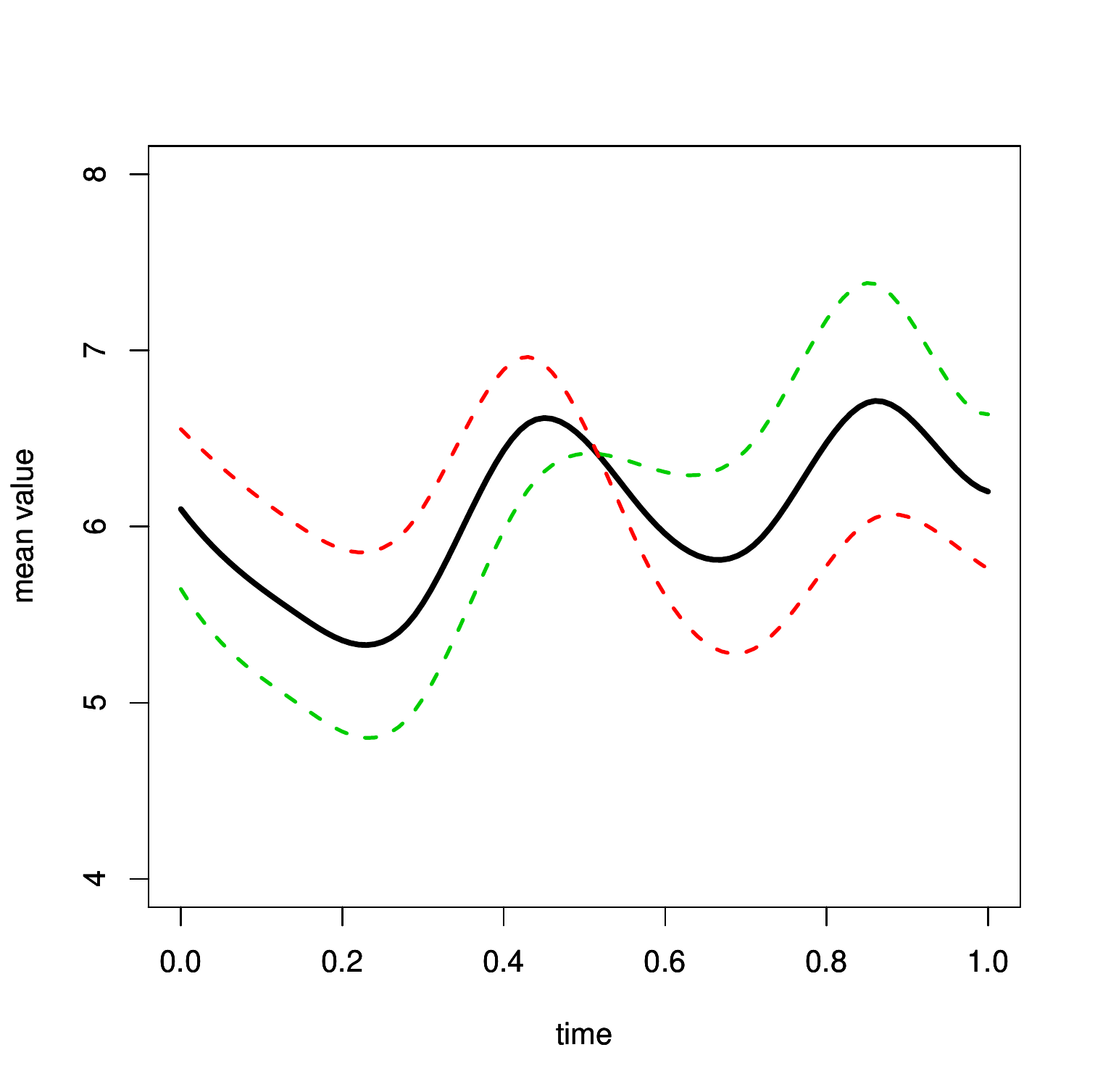} \quad
\includegraphics[width=7cm]{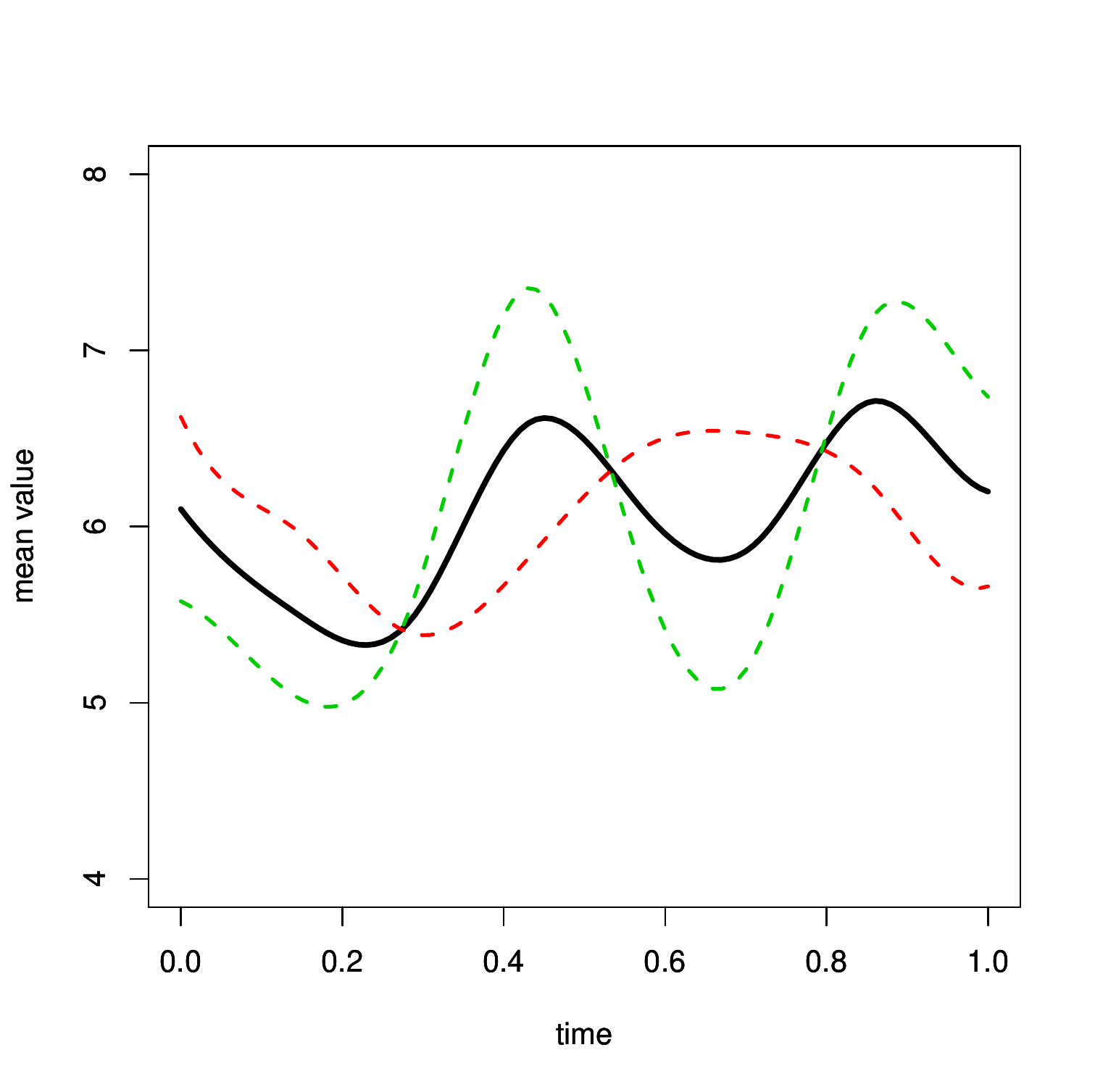}
\caption{Square-root transformed {\tt PM10} observations with fat overall mean curve (upper left panel), effect of the first FPC (upper right panel), effect of the second FPC (lower left panel),  and effect of the third FPC (lower right panel).}
\label{fig:fd}
\end{center}
\end{figure}

For the comparison of the quality of the competing prediction methods, the following was adopted. First, five blocks of consecutive functional observations $Y_{k+1},\ldots,Y_{k+100}$ were chosen, with $k=0,15,30,45,60$. Each block was then used to estimate parameters and fit a certain model. Then, out-of-sample predictions for the values of $Y_{k+100+\ell}$, $\ell=1,\ldots,15$, were made. Finally, the resulting squared prediction errors
\[
\int_0^1\big[Y_{k+100+\ell}(t)-{\tt Pr}(Y_{k+100+\ell})(t)\big]^2dt, \qquad \ell=1,\ldots,15,
\]
were computed, where ${\tt Pr}$ can stand for any of the prediction methods tested. From the $15$ resulting numbers, median ($\mathrm{MED}_{\tt Pr}$) and mean ($\mathrm{MSE}_{\tt Pr}$) were computed.  Results are reported in Table~\ref{tab:pm10}. With the exception of the first period ($k=0$), MSE and MED obtained from the new method are significantly smaller than the ones resulting from the BKR method. In fact, during the second and third period ($k=15$ and $k=30$) prediction errors are on average only about half as big as the ones obtained via BKR. This may arguably be due to an underestimation of the order by BKR method (as evidenced in the simulations).
\begin{table}[ht]
\begin{center}
\begin{tabular}{r|rrr|rrr|rrr|rrr}
  \hline
$k$& $p_a$ & $p_b$ &$ p_c$ & $d_a$ &  $d_b$ & $d_c$ & $\text{MSE}_a$ & $\text{MSE}_b$ &  $\text{MSE}_c$ & $\text{MED}_a$ & $\text{MED}_b$   & $\text{MED}_c$ \\ 
  \hline
 0 & 1 & 1 & 2 & 3 & 3 & 3 & 1.33 & {\bf 1.28} &         1.32 & 1.28 & 1.23 & {\bf 0.88} \\ 
15 &   3 & 1 & 3 & 3 & 3 & 3 & 2.69 &         5.23 & {\bf 2.50} & 2.38 & 5.34 & {\bf 1.45} \\ 
30 &   4 & 1 & 3 & 3 & 2 & 3 & 2.05 &         4.05 & {\bf 1.93} & 1.33 & 2.56 & {\bf 1.26} \\ 
45&   3 & 1 & 3 & 3 & 2 & 3 & 2.25 &         2.44 & {\bf 1.83} & 1.34 & 1.67 & {\bf 1.14 }\\ 
60&   2 & 1 & 1 & 3 & 2 & 5 & 1.22 &         1.82 & {\bf 1.05} & 1.12 & 1.60 & {\bf 0.89} \\ 
   \hline
\end{tabular}
\caption{Comparison of the 3 prediction methods. Subscript $a$ ($b$, $c$) corresponds to method FPE (BKR, FPEX). We report mean (MSE) and median (MED) of the 15 predictions from each block as well as the values of $d$ and $p$ chosen by the respective methods.}
\label{tab:pm10}
\end{center}
\end{table}

PM10 concentrations are known to be high at locations suffering from severe temperature inversions such as the basin areas of the Alps. Following \cite{shp08}, temperature difference between Graz ($350m$ above sea level) and Kalkleiten ($710m$ above sea level) can be utilized to model this phenomenon. Temperature inversion is often seen as a key factor influencing PM10 concentrations because temperatures increasing with sea level result in a sagging exchange of air, thereby yielding a higher pollutant load at the lower elevation. 

To illustrate functional prediction with covariates, temperature difference curves of Graz and Kalkleiten have been included as a dependent variable. For the overall sample, the first two FPCs of the temperature difference curves describe about $92\%$ of the variance. Hence, FPCA was used for covariate dimension reduction, leading to the inclusion of a two-dimensional exogenous regressor (which is almost equivalent to the true regressor curve) in the second step of Algorithm~\ref{alg:3}. Then a $d$-variate VARX($p$) model was fit with $d$ and $p$ selected by the functional final prediction error-type criterion adjusted for the covariate:
\begin{equation}
\label{eq:fFPEX}
\mathrm{fFPE}(p,d)=\frac{n+pd+r}{n-pd-r}\,\mathrm{tr}(\hat\Sigma_{\boldsymbol{Z}})+\sum_{\ell>d}\hat\lambda_\ell.
\end{equation}
Here $r$ is the dimension of the regressor vector (in the present case, $r=2$) and $\hat\Sigma_{\boldsymbol{Z}}$ is the covariance matrix of the residuals when a model of order $p$ and dimension $d$ is fit. The latter method is referred to as FPEX. The corresponding prediction results are summarized in Table~\ref{tab:pm10}. A further significant improvement in the mean and median square (out-of-sample) prediction error can be observed.

\section{Conclusions}
\label{s:con}

This paper proposes a new prediction methodology for functional time series that appears to be widely and easily applicable. It is based on the idea that dimension reduction with functional principal components analysis should lead to a vector-valued time series of FPC scores that can be predicted with any existing multivariate methodology, parametric and nonparametric. The multivariate prediction is then transformed to a functional prediction using a truncated Karhunen-Lo\'eve decomposition. 

The proposed methodology seems to be advantageous for several reasons. Among them is its intuitive appeal, made rigorous for the predominant FAR($p$) case, but also its ease of application as existing software packages can be readily used, even by non-experts. It is in particular straightforward to extend the procedure to include exogenous covariates into the prediction algorithm. Simulations and an application to pollution data suggest that the proposed method leads to predictions that are always competitive with and often superior to the benchmark predictions in the field. 

It is hoped that the present article can spawn interest among researchers working in the active area of functional time series.

\appendix

\section{Theoretical considerations}\label{s:comp}

It is stated in Section~\ref{s:pr} that empirical mean and covariance are $\sqrt{n}$-consistent estimators for their population counterparts for a large class of functional time series. The following lemma makes this statement precise for FAR($p$) processes. The notation of Section~\ref{ss:comp} is adopted.
\begin{lemma}\label{le:A1}
Consider the FAR($p$) model \eqref{e:farp} and suppose that Assumption~{\tt FAR} holds.  Further suppose that  $\|(\Psi^*)^{k_0}\|_\mathcal{L}<1$ for some $k_0\geq 1$. Then (i) $E[\|\hat \mu_n-\mu\|^2]=O(1/n)$. (ii) If in addition $(\varepsilon_k)$ in $L^4_H$, then $E[\|\hat C_n-C\|^2]=O(1/n)$.
\end{lemma}
\begin{proof}
If follows from Proposition~2.1 in \cite{hk11} and Theorem~3.1 in \cite{b00} that $(X_k)$ is $L^2$-$m$-approximable under (i) and $L^4$-$m$-approximable under (ii). $L^p$-$m$-approximability is inherited by the projection $\pi(X_k)=X_k^{(1)}=Y_k$. Now the proof follows from Theorems~5 and 6 in \cite{hk12}.
\end{proof}

\subsection{The VAR structure}\label{ss:compv_f}

In case of a VAR(1), Step 2.\ of Algorithm \ref{alg:1} can be performed with least squares. To explicitly calculate $\hat{\bY}^e_{n+1}$, apply $\langle \cdot,\hv_\ell\rangle$ to both sides of $Y_k=\Psi(Y_{k-1})+\varepsilon_k$ to obtain
\begin{align}
\langle Y_k,\hv_\ell\rangle&=\langle \Psi(Y_{k-1}),\hv_\ell\rangle+\langle \varepsilon_k,\hv_\ell\rangle \nonumber\\
&=\sum_{\ell^\prime=1}^\infty\langle Y_{k-1},\hv_{\ell^\prime}\rangle\langle \Psi(\hv_{\ell^\prime}),\hv_\ell\rangle+\langle \varepsilon_k,\hv_\ell\rangle \nonumber\\
&=\sum_{\ell^\prime=1}^d\langle Y_{k-1},\hv_{\ell^\prime}\rangle\langle \Psi(\hv_{\ell^\prime}),\hv_\ell\rangle+\delta_{k,\ell}, \label{stack}
\end{align}
with remainder terms $\delta_{k,\ell}=d_{k,\ell}+\langle \varepsilon_k,\hv_\ell\rangle$
where
\begin{align*}
d_{k,\ell}&=\sum_{\ell^\prime=d+1}^\infty\langle Y_{k-1},\hat v_{\ell^\prime}\rangle\langle \Psi(\hat v_{\ell^\prime}),\hv_\ell\rangle,
\end{align*}
noting that $(\hv_\ell)$ can always be extended to an orthonormal basis of $L^2$. Some notation is needed. Set
$
\be_k=(\langle\varepsilon_k,v_1\rangle,\ldots,\langle\varepsilon_k,v_d\rangle)^\prime
$
and
$
\bu_k=(u_{k,1},\ldots,u_{k,d})^\prime,
$
where
$
u_{k,\ell}=\sum_{\ell^\prime>d}\langle Y_{k-1},v_{\ell^\prime}\rangle\langle
\Psi(v_{\ell^\prime}),v_\ell\rangle,
$
and let $B_d\in\mathbb{R}^{d\times d}$ be the matrix with entry
$\langle\Psi(v_\ell),v_{\ell^\prime}\rangle$ in the $\ell$th row and the $\ell^\prime$th column, $\ell,\ell^\prime=1,\ldots,d$. Let moreover $\bbeta=\mathrm{vec}(B_d^\prime)$,
$\bZ=(\bY_2^\prime,\ldots,\bY_n^\prime)^\prime$, $\bE=(\be_2^\prime,\ldots,\be_n^\prime)^\prime$, $\bU=(\bu_2^\prime,\ldots,\bu_n^\prime)^\prime$, $X_k=I_d\otimes \bY_k^\prime$ and $X=(X_1^\prime\colon\ldots\colon X_{n-1}^\prime)^\prime$.
Replacing the eigenfunctions $v_\ell$ by their sample counterparts $\hat v_\ell$, empirical versions of the above variables are denoted by $\bY_k^e$, $\bZ^e$, $X_k^e$, $X^e$, $B_d^e$ and $\bbeta_d^e$. For a vector $\mathbf{x}\in\mathbb{R}^{d^2}$, the operation $\mathrm{mat}(\mathbf{x})$
creates a $d\times d$ matrix, whose $\ell$-th column contains the elements $v_{(1-\ell)d+1},\ldots,v_{\ell d}$.
Define now $\bdelta_k=(\delta_{k,1},\ldots,\delta_{k,d})^\prime$ to arrive at the equations
\begin{equation}\label{emp_var}
\bY_k^e=B_d^e\,\bY_{k-1}^e+\bdelta_k,\qquad k=2,\ldots, n.
\end{equation}
The equations in \eqref{emp_var} formally resemble VAR(1) equations. Notice, however,
that it is a nonstandard formulation, since the errors $\bdelta_k$  are generally not centered and dependent. Furthermore, $\bdelta_k$ depends in a complex way on $\bY_{k-1}^e$, so that the errors are not uncorrelated with past observations. The coefficient matrix $B_d^e$ is also random, but fixed for fixed sample size $n$. In the sequel these effects are ignored. Utilizing some matrix algebra, \eqref{emp_var} can be written as the linear regression
\begin{equation}\label{reg}
\bZ^e=X^e\bbeta_d^e+\bDelta,
\end{equation}
where $\bDelta=(\bdelta_2^\prime,\ldots,\bdelta_n^\prime)^\prime$.
The ordinary least squares estimator is then
$
\hat{\bbeta}^e_d=({X^e}^\prime X^e)^{-1}{X^e}^\prime\bZ^e,
$
and the prediction equation
\begin{equation}\label{pred_eq}
\hat\bY^e_{n+1}=\hat B_d^e\bY_n^e=(\hat{y}_{n+1,1}^e,\ldots,\hat{y}_{n+1,d}^e)^\prime,
\end{equation}
follows directly, defining $\hat B_d^e=\mathrm{mat}\big(\hat{\bbeta}_d^e\big)^\prime$.

\subsection{Proof of Theorem~\ref{th:equiv_VAR_FAR}}\label{ss:compv_f-b}

Recall the notations introduced above equation \eqref{emp_var}. In order to prove the asymptotic equivalence between $\tilde Y_{n+1}$ in \eqref{fc_kl} and $\hat Y_{n+1}$ in \eqref{fc_standard} for the case of FAR(1) functional time series, observe first that
\[
\bigg(\frac{1}{n-1}{X^e}^\prime X^e\bigg)^{-1}=I_d\otimes\hat\Gamma^{-1},
\]
where $\hat\Gamma$ is the $d\times d$ matrix with entries $\hat\Gamma(\ell,\ell^\prime)=\frac{1}{n-1}\sum_{k=1}^{n-1}y_{k,\ell}^ey_{k,\ell^\prime}^e$ determined by the FPC scores $y_{k,\ell}^e=\langle Y_{k},\hv_\ell\rangle$, and $\otimes$ signifies the Kronecker product. With the help of \eqref{pred_eq}, the VAR(1) based predictor \eqref{fc_kl} can be written in the form
\[
\hat Y_{n+1}=
\frac{1}{n-1}\left\{\left(\mathrm{mat}\left(\big[I_d\otimes\hat\Gamma^{-1}\big]{X^e}^\prime\bZ^e\right)\right)^\prime\bY_n^e\right\}^\prime\hat{\boldsymbol{v}},
\]
with $\hat{\boldsymbol{v}}=(\hv_1,\ldots,\hv_d)^\prime$ being the vector of the first $d$ empirical eigenfunctions. On the other hand, defining the $d\times d$ matrix $\tilde\Gamma$ by the entries $\tilde\Gamma(\ell,\ell^\prime)=\frac{1}{n}\sum_{k=1}^{n}y_{k,\ell}^ey_{k,\ell^\prime}^e=\mathrm{diag}(\hat\lambda_1,\ldots,\hat\lambda_d)$, direct verification shows that \eqref{fc_standard} takes the form
\[
\tilde{Y}_{n+1}=
\frac{1}{n-1}\left\{\left(\mathrm{mat}\left(\big[I_d\otimes\tilde\Gamma^{-1}\big]{X^e}^\prime\bZ^e\right)\right)^\prime\bY_n^e\right\}^\prime\hat{\boldsymbol{v}}.
\]
The only formal difference between the two predictors under consideration is therefore in the matrices $\hat\Gamma$ and $\tilde\Gamma$. Now, for any $\ell,\ell^\prime=1,\ldots,d$,
\begin{align*}
\hat\Gamma(\ell,\ell^\prime)
&=\tilde\Gamma(\ell,\ell^\prime)+\frac{1}{n-1}\frac{1}{n}\sum_{k=1}^{n}y_{k,\ell}^ey_{k,\ell^\prime}^e
-\frac{1}{n-1}y_{n,\ell}^ey_{n,\ell^\prime}^e \\
&=\tilde\Gamma(\ell,\ell^\prime)
+\frac{1}{n-1}\left(\hat\lambda_\ell I\{\ell=\ell^\prime\}-y_{n,\ell}^ey_{n,\ell^\prime}^e\right),
\end{align*}
so that $Y_n\in L_H^2$ implies
\[
\left|\hat\Gamma(\ell,\ell^\prime)-\tilde\Gamma(\ell,\ell^\prime)\right|
\leq\frac{1}{n-1}\bigg(\frac{1}{n}\sum_{k=1}^n\|Y_k\|^2 +\|Y_n\|^2\bigg)
=O_P\bigg(\frac 1n\bigg).
\]
In the following $\|\cdot\|$ will be used for the $L^2$ norm, the Euclidean norm in $\mathbb{R}^d$ and matrix norm $\|A\|=\sup_{\|\mathbf{x}\|=1}\|A\mathbf{x}\|$, for a square matrix $A\in\mathbb{R}^{d\times d}$. Let 
\[
\Delta=\mathrm{mat}\bigg(\big[I_d\otimes\big(\hat\Gamma^{-1}-\tilde\Gamma^{-1}\big)\big]\frac{1}{n-1}{X^e}^\prime\bZ^e\bigg).
\]
The orthogonality of the $\hv_\ell$ together with Pythagoras' theorem and Bessel's inequality imply that
\begin{align*}
\|\hat{Y}_{n+1}-\tilde{Y}_{n+1}\|&
=\left\|\Delta^\prime\bY_n^e\right\|
\leq \|\Delta\|\|\bY_n^e\|
= \|\Delta\|\left(\sum_{\ell=1}^d(y_{n,\ell}^e)^2\right)^{1/2}
\leq \|\Delta\|\|Y_n\|.
\end{align*}
Define $S=\mathrm{mat}(\frac{1}{n-1}{X^e}^\prime\bZ^e)$ and notice that
$\Delta=(\hat\Gamma^{-1}-\tilde\Gamma^{-1})S$ and hence $
\|\Delta\|\leq\big\|\hat\Gamma^{-1}-\tilde\Gamma^{-1}\big\|\|S\|$.

Let $\boldsymbol{w}=(w_1,\ldots,w_d)^\prime$. Since
$S(\ell,\ell^\prime)=\frac{1}{n-1}\sum_{k=1}^{n-1}y_{k,\ell}^ey_{k+1,\ell^\prime}^e$, iterative applications of the Cauchy-Schwarz inequality yield
\begin{align*}
\|S\|^2&=\sup_{\|\boldsymbol{w}\|=1} \sum_{\ell=1}^d\bigg(\sum_{\ell^\prime=1}^d\frac{1}{n-1}\sum_{k=1}^{n-1}y_{k,\ell}^ey_{k+1,\ell^\prime}^e w_{\ell^\prime}\bigg)^2\\
&\leq \sum_{\ell=1}^d\sum_{\ell^\prime=1}^d\bigg(\frac{1}{n-1}\sum_{k=1}^{n-1}y_{k,\ell}^ey_{k+1,\ell^\prime}^e\bigg)^2\\
&\leq \sum_{\ell=1}^d\sum_{\ell^\prime=1}^d\frac{1}{n-1}\sum_{k=1}^n(y_{k,\ell}^e)^2
\frac{1}{n-1}\sum_{k=1}^n(y_{k,\ell^\prime}^e)^2\\
&\leq \bigg(\frac{1}{n-1}\sum_{k=1}^n\|Y_k\|^2\bigg)^2\\
&=O_P(1).
\end{align*}
It remains to estimate $\|\hat\Gamma^{-1}-\tilde\Gamma^{-1}\|$. The next step consists of using the fact that, for any $A,B\in \mathbb{R}^{d\times d}$, it holds that $(A+B)^{-1}=A^{-1}-A^{-1}(I+BA^{-1})^{-1}BA^{-1}$, provided all inverse matrices exist. Now choose $A=\tilde\Gamma$ and $B=\hat\Gamma-\tilde\Gamma$. Since in the given setting the time series $(Y_n)$ is stationary and ergodic, it can be deduced that $\hat\lambda_d\to\lambda_d$ with probability one. Thus $\hat\lambda_d^{-1}\|\tilde\Gamma-\hat\Gamma\|<1$ for large enough $n$,  and consequently
\begin{align*}
\big\|\hat\Gamma^{-1}-\tilde\Gamma^{-1}\big\|
&=\Big\|\tilde\Gamma^{-1}\big[I_d+(\hat\Gamma-\tilde\Gamma)
\tilde\Gamma^{-1}\big]^{-1}(\hat\Gamma-\tilde\Gamma)
\tilde\Gamma^{-1}\Big\| \\[.2cm]
&\leq\big\|\tilde\Gamma^{-1}\big\|^2\big\|\hat\Gamma-\tilde\Gamma\big\|\Big\|\big[I_d+(\hat\Gamma-\tilde\Gamma)\tilde\Gamma^{-1}\big]^{-1}\Big\| \\[.2cm]
&\leq\frac{\big\|\tilde\Gamma-\hat\Gamma\big\|}{\hat\lambda_d^2}\sum_{\ell=0}^\infty\bigg(\frac{\|\tilde\Gamma-\hat\Gamma\|}{\hat\lambda_d}\bigg)^{\ell}\\
&=O_P\bigg(\frac 1n\bigg).
\end{align*}
It has been assumed here that $\lambda_d>0$. If $\lambda_d=0$, then the model has dimension $d'<d$. In this case both estimators will of course be based on at most $d'$ principal components.

Putting together all results, the statement of Theorem \ref{th:equiv_VAR_FAR} is established.

\subsection{Proof of Theorem \ref{th:prederror}}
\label{proof:th3.2}
Using the results and notations of Section~\ref{ss:fpe}, it follows that
\begin{align*}
E\big[\|Y_{n+1}-\hat Y_{n+1}\|^2\big]
&=E\big[\|\boldsymbol{Y}_{n+1}-\hat{\boldsymbol{Y}}_{n+1}\|^2\big]
+\sum_{i>d}\lambda_i.
\end{align*}
Some algebra shows that\
\begin{align*}
\boldsymbol{Y}_{n+1}=\boldsymbol\Psi_1 \boldsymbol{Y}_{n}+\cdots+\boldsymbol\Psi_p\boldsymbol{Y}_{n-p+1}+\boldsymbol{E}_n,
\end{align*}
where the $d\times d$ matrices $\boldsymbol\Psi_j$ have entry $\langle\Psi_j(v_{\ell^\prime}),v_\ell\rangle$ in the $\ell^\prime$th column and $\ell$th row, and $\boldsymbol{E}_n=\boldsymbol{T}_n+\boldsymbol{S}_n$ with $d$-variate vectors 
$\boldsymbol{T}_n$ and $\boldsymbol{S}_n$ taking the respective values $\sum_{j=1}^p\sum_{\ell^\prime=d+1}^\infty y_{n+1-j,\ell^\prime}\langle \Psi_j(v_{\ell^\prime}),v_{\ell}\rangle$ and $\langle\varepsilon_{n+1},v_\ell\rangle$ in the $\ell$th coordinate. 

The best linear predictor $\hat{\boldsymbol{Y}}_{n+1}$ of $\boldsymbol{Y}_{n+1}$ based on $\boldsymbol{Y}_{1},\ldots,\boldsymbol{Y}_{n}$ satisfies
\[
E\big[\|\boldsymbol{Y}_{n+1}-\hat{\boldsymbol{Y}}_{n+1}\|^2\big]
\leq E\big[\|\boldsymbol{Y}_{n+1}-(\boldsymbol\Psi_1 \boldsymbol{Y}_{n}+\cdots+\boldsymbol\Psi_p\boldsymbol{Y}_{n-p+1})\|^2\big]
=E\big[\|\boldsymbol{S}_n\|^2\big]+E\big[\|\boldsymbol{T}_n\|^2\big].
\]
The last equality comes from the fact that, due to causality, the components in $\boldsymbol{S}_n$ and in $\boldsymbol{T}_n$ are uncorrelated. Observe next that, by Bessel's inequality,
$
E[\|\boldsymbol{S}_n\|^2]
=\sum_{\ell=1}^dE[\langle \varepsilon_{n+1},v_\ell\rangle^2]\leq\sigma^2.
$
It remains to bound $E[\|\boldsymbol{T}_n\|^2]$. For this term, it holds
\begin{align*}
E\big[\|\boldsymbol{T}_n\|^2\big]
&=E\Bigg[\sum_{\ell=1}^d\bigg(\sum_{j=1}^p\sum_{\ell^\prime=d+1}^\infty y_{n+1-j,\ell^\prime}\langle \Psi_j(v_{\ell^\prime}),v_\ell\rangle\bigg)^2\Bigg] \\
&\leq E\Bigg[\sum_{d=1}^\infty\bigg\langle\sum_{j=1}^{p}\sum_{\ell^\prime=d+1}^\infty y_{n+1-j,\ell^\prime} \Psi_j(v_{\ell^\prime}),v_\ell\bigg\rangle^2\Bigg] \\
&=E\Bigg[\bigg\|\sum_{j=1}^p\sum_{\ell^\prime=d+1}^\infty y_{n+1-j,\ell^\prime} \Psi_j(v_{\ell^\prime})\bigg\|^2\Bigg],
\end{align*}
where Parseval's identity was applied in the final step. Repeatedly using the Cauchy-Schwarz inequality, the last expectation can be estimated as
\begin{align*}
\sum_{j,j^\prime=1}^p&\sum_{\ell,\ell^\prime=d+1}^\infty 
E\big[y_{n+1-j,\ell}\,y_{n+1-j^\prime,\ell^\prime}\big]
\big\langle \Psi_{j}(v_{\ell}),\Psi_{j^\prime}(v_{\ell^\prime})\big\rangle \\
&\leq\sum_{j,j^\prime=1}^p
\bigg(\sum_{\ell=d+1}^\infty\sqrt{\lambda_{\ell}}\| \Psi_{j}(v_{\ell})\|\bigg)
\bigg(\sum_{\ell^\prime=d+1}^\infty\sqrt{\lambda_{\ell^\prime}}\| \Psi_{j^\prime}(v_{\ell^\prime})\|\bigg)\\
&\leq\sum_{\ell''=d+1}^\infty\lambda_{\ell''}\sum_{j,j^\prime=1}^p
\bigg(\sum_{\ell=d+1}^\infty\|\Psi_{j}(v_{\ell})\|^2\bigg)^{1/2}
\bigg(\sum_{\ell^\prime=d+1}^\infty\|\Psi_{j^\prime}(v_{\ell^\prime})\|^2\bigg)^{1/2} \\
&=\sum_{\ell=d+1}^\infty\lambda_\ell\bigg(\sum_{j=1}^p\bigg[\sum_{\ell=d+1}^\infty\|\Psi_{j}(v_{\ell})\|^2\bigg]^{1/2}\bigg)^2.
\end{align*}
Collecting all estimates finishes the proof.

\section*{Acknowledgment} 
Alexander Aue is Associate Professor, Department of Statistics, University of California, Davis, CA 95616 (E-mail: {\tt aaue@ucdavis.edu}). 
Diogo Dubart Norinho is graduate student, Department of Computer Science, University College London, London WC1E 6BT, UK (E-mail: {\tt ucabdub@ucl.ac.uk}.
Siegfried H\"ormann is Charg\'e de cours, Department of Mathematics, Universit\'e libre de Bruxelles, B-1050 Brussels, Belgium (E-mail: {\tt shormann@ulb.ac.be}).  
The authors are grateful to the editor, the associate editor and the reviewers for constructive comments and support. 
This research was partially supported by NSF grants DMS 0905400, DMS 1209226 and DMS 1305858, Communaut\'e fran\c{c}aise de Belgique---Actions de Recherche Concert\'ees (2010--2015) and the IAP research network grant nr.\ P7/06 of the Belgian government (Belgian Science Policy).
Part of this work was completed during a Research in Pairs stay (Aue and H\"ormann) at the Mathematical Research Institute Oberwolfach.

\vspace{.3cm}
\renewcommand{\baselinestretch}{.1}

\end{document}